\documentclass[a4paper]{article}
\usepackage{amsfonts}
\usepackage{amsmath}
\usepackage{amssymb}
\usepackage{amsbsy}
\usepackage{mathrsfs}
\usepackage{dsfont}
\usepackage{blindtext}
\usepackage{amsthm}
\usepackage{enumitem}
\usepackage{diagrams}

\title{\Large \bf Intersection Types for the Computational $\lambda$-Calculus}
\author{Ugo de'Liguoro\\[1mm] \small Dipartimento di Informatica\\ 
\small Universit\`a di Torino\\ 
\small {\tt ugo.deliguoro@unito.it}  \and 
Riccardo Treglia \\[1mm] 
\small Dipartimento di Informatica\\ 
\small Universit\`a di Torino\\ 
\small {\tt riccardo.treglia@unito.it}}

\newtheorem{definition}{Definition}[section]

\newtheorem{theorem}[definition]{Theorem}
\newtheorem{corollary}[definition]{Corollary}
\newtheorem{lemma}[definition]{Lemma}
\newtheorem{proposition}[definition]{Proposition}
\newtheorem{remark}[definition]{Remark}
\newtheorem{example}[definition]{Example}

\newdimen\proofrulebreadth \proofrulebreadth=.05em
\newdimen\proofdotseparation \proofdotseparation=1.25ex
\newdimen\proofrulebaseline \proofrulebaseline=2ex
\newcount\proofdotnumber \proofdotnumber=3
\let\then\relax
\def\hfi{\hskip0pt plus.0001fil}
\mathchardef\squigto="3A3B
\newif\ifinsideprooftree\insideprooftreefalse
\newif\ifonleftofproofrule\onleftofproofrulefalse
\newif\ifproofdots\proofdotsfalse
\newif\ifdoubleproof\doubleprooffalse
\let\wereinproofbit\relax
\newdimen\shortenproofleft
\newdimen\shortenproofright
\newdimen\proofbelowshift
\newbox\proofabove
\newbox\proofbelow
\newbox\proofrulename
\def\shiftproofbelow{\let\next\relax\afterassignment\setshiftproofbelow\dimen0 }
\def\shiftproofbelowneg{\def\next{\multiply\dimen0 by-1 }%
\afterassignment\setshiftproofbelow\dimen0 }
\def\setshiftproofbelow{\next\proofbelowshift=\dimen0 }
\def\setproofrulebreadth{\proofrulebreadth}

\def\prooftree{
\ifnum  \lastpenalty=1
\then   \unpenalty
\else   \onleftofproofrulefalse
\fi
\ifonleftofproofrule
\else   \ifinsideprooftree
        \then   \hskip.5em plus1fil
        \fi
\fi
\bgroup
\setbox\proofbelow=\hbox{}\setbox\proofrulename=\hbox{}%
\let\justifies\proofover\let\leadsto\proofoverdots\let\Justifies\proofoverdbl
\let\using\proofusing\let\[\prooftree
\ifinsideprooftree\let\]\endprooftree\fi
\proofdotsfalse\doubleprooffalse
\let\thickness\setproofrulebreadth
\let\shiftright\shiftproofbelow \let\shift\shiftproofbelow
\let\shiftleft\shiftproofbelowneg
\let\ifwasinsideprooftree\ifinsideprooftree
\insideprooftreetrue
\setbox\proofabove=\hbox\bgroup$\displaystyle 
\let\wereinproofbit\prooftree
\shortenproofleft=0pt \shortenproofright=0pt \proofbelowshift=0pt
\onleftofproofruletrue\penalty1
}

\def\eproofbit{
\ifx    \wereinproofbit\prooftree
\then   \ifcase \lastpenalty
        \then   \shortenproofright=0pt  
        \or     \unpenalty\hfil         
        \or     \unpenalty\unskip       
        \else   \shortenproofright=0pt  
        \fi
\fi
\global\dimen0=\shortenproofleft
\global\dimen1=\shortenproofright
\global\dimen2=\proofrulebreadth
\global\dimen3=\proofbelowshift
\global\dimen4=\proofdotseparation
\global\count255=\proofdotnumber
$\egroup  
\shortenproofleft=\dimen0
\shortenproofright=\dimen1
\proofrulebreadth=\dimen2
\proofbelowshift=\dimen3
\proofdotseparation=\dimen4
\proofdotnumber=\count255
}

\def\proofover{
\eproofbit 
\setbox\proofbelow=\hbox\bgroup 
\let\wereinproofbit\proofover
$\displaystyle
}%
\def\proofoverdbl{
\eproofbit 
\doubleprooftrue
\setbox\proofbelow=\hbox\bgroup 
\let\wereinproofbit\proofoverdbl
$\displaystyle
}%
\def\proofoverdots{
\eproofbit 
\proofdotstrue
\setbox\proofbelow=\hbox\bgroup 
\let\wereinproofbit\proofoverdots
$\displaystyle
}%
\def\proofusing{
\eproofbit 
\setbox\proofrulename=\hbox\bgroup 
\let\wereinproofbit\proofusing
\kern0.3em$
}

\def\endprooftree{
\eproofbit 
  \dimen5 =0pt
\dimen0=\wd\proofabove \advance\dimen0-\shortenproofleft
\advance\dimen0-\shortenproofright
\dimen1=.5\dimen0 \advance\dimen1-.5\wd\proofbelow
\dimen4=\dimen1
\advance\dimen1\proofbelowshift \advance\dimen4-\proofbelowshift
\ifdim  \dimen1<0pt
\then   \advance\shortenproofleft\dimen1
        \advance\dimen0-\dimen1
        \dimen1=0pt
        \ifdim  \shortenproofleft<0pt
        \then   \setbox\proofabove=\hbox{%
                        \kern-\shortenproofleft\unhbox\proofabove}%
                \shortenproofleft=0pt
        \fi
\fi
\ifdim  \dimen4<0pt
\then   \advance\shortenproofright\dimen4
        \advance\dimen0-\dimen4
        \dimen4=0pt
\fi
\ifdim  \shortenproofright<\wd\proofrulename
\then   \shortenproofright=\wd\proofrulename
\fi
\dimen2=\shortenproofleft \advance\dimen2 by\dimen1
\dimen3=\shortenproofright\advance\dimen3 by\dimen4
\ifproofdots
\then
        \dimen6=\shortenproofleft \advance\dimen6 .5\dimen0
        \setbox1=\vbox to\proofdotseparation{\vss\hbox{$\cdot$}\vss}%
        \setbox0=\hbox{%
                \advance\dimen6-.5\wd1
                \kern\dimen6
                $\vcenter to\proofdotnumber\proofdotseparation
                        {\leaders\box1\vfill}$%
                \unhbox\proofrulename}%
\else   \dimen6=\fontdimen22\the\textfont2 
        \dimen7=\dimen6
        \advance\dimen6by.5\proofrulebreadth
        \advance\dimen7by-.5\proofrulebreadth
        \setbox0=\hbox{%
                \kern\shortenproofleft
                \ifdoubleproof
                \then   \hbox to\dimen0{%
                        $\mathsurround0pt\mathord=\mkern-6mu%
                        \cleaders\hbox{$\mkern-2mu=\mkern-2mu$}\hfill
                        \mkern-6mu\mathord=$}%
                \else   \vrule height\dimen6 depth-\dimen7 width\dimen0
                \fi
                \unhbox\proofrulename}%
        \ht0=\dimen6 \dp0=-\dimen7
\fi
\let\doll\relax
\ifwasinsideprooftree
\then   \let\VBOX\vbox
\else   \ifmmode\else$\let\doll=$\fi
        \let\VBOX\vcenter
\fi
\VBOX   {\baselineskip\proofrulebaseline \lineskip.2ex
        \expandafter\lineskiplimit\ifproofdots0ex\else-0.6ex\fi
        \hbox   spread\dimen5   {\hfi\unhbox\proofabove\hfi}%
        \hbox{\box0}%
        \hbox   {\kern\dimen2 \box\proofbelow}}\doll%
\global\dimen2=\dimen2
\global\dimen3=\dimen3
\egroup 
\ifonleftofproofrule
\then   \shortenproofleft=\dimen2
\fi
\shortenproofright=\dimen3
\onleftofproofrulefalse
\ifinsideprooftree
\then   \hskip.5em plus 1fil \penalty2
\fi
}

\newcommand{\comp}{\circ}  
\newcommand{\Id}{\mbox{\rm id}}

\newcommand{\Dom}[1]{\mbox{\rm dom} \, #1}

\newcommand{\Dcat}{{\cal D}}

\newcommand{\Sub}[2]{\mbox{\rm Sub}_{#1}(#2)}

\newcommand{\Set}[1]{\{#1\}}
\newcommand{\Pair}[2]{\langle #1, #2 \rangle}
\newcommand{\Power}{\mathscr{P}} 

\newcommand{\Iff}{\Leftrightarrow}
\newcommand{\Then}{\Rightarrow}

\newcommand{\Proof}{\noindent {\bf Proof. }}
\def \QED {~\hfill\hbox{\rule{5.5pt}{5.5pt}\hspace*{.5\leftmarginii}}}

\newcommand{\der}{\vdash}
\newcommand{\UnitSub}[1]{\mbox{\it unit}_{#1}}
\newcommand{\Unit}{\mbox{\it unit}\;}

\newcommand{\Bind}{\star}
\newcommand{\Sem}[1]{[\hspace{-0.6mm}[ #1 ]\hspace{-0.6mm}]}

\newcommand{\Let}[3]{\mbox{\it let $#1 = #2$ in $#3$}}

\newcommand{\Var}{\textit{Var}}

\newcommand{\Val}{\textsf{V}}
\newcommand{\Comp}{\textsf{C}}
\newcommand{\TFlambdaComp}{\lambda_c^u}

\newcommand{\Subst}[2]{[ #1 / #2]}
\newcommand{\Env}{\mbox{\it Term-Env}}
\newcommand{\FV}{\textit{FV}}

\newcommand{\metalambda}{%
	\mathop{%
		\rlap{$\lambda$}%
		\mkern2mu
		\raisebox{.275ex}{$\lambda$}%
	}%
}
\makeatletter
\newcommand{\circlearrow}{}
\DeclareRobustCommand{\circlearrow}{%
	\mathrel{\vphantom{\rightarrow}\mathpalette\circle@arrow\relax}%
}
\newcommand{\circle@arrow}[2]{%
	\m@th
	\ooalign{%
		\hidewidth$#1\circ\mkern1mu$\hidewidth\cr
		$#1\longrightarrow$\cr}%
}
\makeatother

\newcommand{\ValTerm}{\textit{Val}\,}
\newcommand{\ComTerm}{\textit{Com}}
\newcommand{\Term}{\textit{Term}}
\newcommand{\Terms}{\textit{Term}}
\newcommand{\Red}{\longrightarrow}
\newcommand{\RedStar}{\stackrel{*}{\Red}}
\newcommand{\Betac}{\beta_c}
\newcommand{\BindLeft}{\textit{ass}}
\newcommand{\BindRight}{\textit{id}}

\newcommand{\MonRed}{\lambda\textbf{C}}

\newcommand{\BetaIdRed}{\Red_{\beta_c, \textit{id}}\,}
\newcommand{\BetaIdRedStar}{\RedStar_{\beta_c, \textit{id}}\,}
\newcommand{\AssRed}{\Red_{\textit{ass}}\,}

\newcommand{\CR}{\textit{CR}}
\newcommand{\WCR}{\textit{WCR}}
\newcommand{\SN}{\textit{SN}}

\newcommand{\TypeVar}{\textit{TypeVar}}
\newcommand{\ValType}{\textit{ValType}}
\newcommand{\ComType}{\textit{ComType}}
\newcommand{\Types}{{\cal T}}
\newcommand{\Th}{\textit{Th}}

\newcommand{\Inter}{\wedge}
\renewcommand{\QED}{ \begin{flushright}$\square$\end{flushright} }

\newcommand{\Filt}{{\cal F}}

\newcommand{\Compact}{{\cal K}}
\newcommand{\CompactOp}{{\cal K}^{\textit{\tiny op}}}
\newcommand{\Order}{\sqsubseteq}
\newcommand{\OrderOp}{\Order^{\textit{\tiny op}}}
\newcommand{\TypeTh}{{\cal T}}
\newcommand{\TT}{\textit{\bf T}}
\newcommand{\Up}{~\uparrow\!}

\newcommand{\TypeEnv}{\textit{TypeEnv}}

\newcommand{\invLim}{\lim_\leftarrow}
\newcommand{\Nat}{\mathds{N}}

\newcommand{\IntrArr}{(\to \mbox{I})}
\newcommand{\ElimArr}{(\to \mbox{E})}
\newcommand{\IntrUnit}{(\Unit \mbox{I})}
\newcommand{\IntrInter}{(\Inter I)}

\newcommand{\Adm}{\textit{Adm}}

\newcommand{\I}{{\cal I}}
\newcommand{\Iinterp}[1]{| #1 |}

\newcommand{\TP}{\textit{TP}}
\newcommand{\DP}{\textit{DP}}

\newcommand{\ValCtx}{\mathcal{V}}
\newcommand{\CompCtx}{ \mathcal{C}}
\newcommand{\Hole}[2]{\langle #1_{#2} \rangle}

\begin{document}
\maketitle
\begin{abstract} 
We study polymorphic type assignment systems for untyped $\lambda$-calculi with effects, based on Moggi's monadic approach.
Moving from the abstract definition of monads, we introduce a version of the call-by-value computational $\lambda$-calculus
based on Wadler's variant with unit and bind combinators, and without let.
We define a notion of reduction for the calculus and prove it confluent, and also we relate our calculus to the original work by Moggi showing that his untyped metalanguage can be interpreted and simulated in our calculus.
We then introduce an intersection type system inspired to Barendregt, Coppo and Dezani system for ordinary untyped $\lambda$-calculus, 
establishing type invariance under conversion, and provide models of the calculus via inverse limit and filter model constructions and relate them. We prove soundness and completeness of the type system, together with subject reduction and expansion properties. 
Finally, we introduce a notion of convergence, which is precisely related to reduction, and characterize convergent terms via their types.

\end{abstract}


\newcommand{\retract}{\triangleleft}

\section{Introduction}\label{sec:introduction}

The computational $\lambda$-calculus was introduced by Moggi \cite{Moggi'89,Moggi'91} as a meta-language to describe non-functional effects in programming languages via an incremental approach. The basic idea is to distinguish among values of some type $D$ and computations over such values, the latter having type $TD$. Semantically $T$ is a monad, endowing $D$ with a richer structure such that operations over computations can be seen as algebras of $T$.

The monadic approach is not just useful when building compilers modularly with respect to various kinds of effects \cite{Moggi'91}, to interpret languages with effects like control operators via a CPS translation \cite{Filinski:popl1994}, or to write effectful programs in a purely functional language such as Haskell \cite{Wadler-Monads}, but also to reason about such programs. In this respect, typed computational lambda-calculus has been related to static program analysis and type and effect systems \cite{BentonHM00}, 
PER based relational semantics \cite{BentonKHB06},
and more recently co-inductive methods for reasoning about effectful programs have been investigated, leading to principles that hold for arbitrary monads and their algebras \cite{LagoGL17}. 

Aim of our work is to investigate the monadic approach to effectful functional languages in the untyped case. This is motivated by the fact that similar, if not more elusive questions arise for effectful untyped languages as well as for typed ones; but also because the untyped setting is the natural one where studying program analysis via type assignment systems in Curry style, like in the case of intersection types, which we advocate. Indeed, in working out the approach in the untyped case lays the foundation for doing the same also for typed languages, either by seeing intersection types as refinement types, or by looking at them as to the formulas of the endogenous logic of domain theoretic interpretations of types \cite{Abramsky'91}.

It might appear nonsense to speak of monads w.r.t. an untyped calculus, as the monad $T$  interprets a type constructor both in Moggi's and in Wadler's 
formulation of the computational $\lambda$-calculus  \cite{Moggi'91,Wadler-Monads}. However, 
much as the untyped $\lambda$-calculus can be seen as a calculus with a single type, which is interpreted by 
a retract of its own function space in a suitable category as formerly observed by Scott \cite{Scott80},
the untyped computational $\lambda$-calculus $\TFlambdaComp$ has two types: the type of values $D$ and the type of computations $TD$. 
The type $D$ is a retract of $D \to TD$, written $D \retract D \to TD$, that is an appropriate space of functions from values to computations \cite{Moggi'89}.
Consequently, we have two sorts of terms, called {\em values} and {\em computations} denoting points in $D$ and $TD$ respectively, and a reduction
relation among computations that turns out to be Church-Rosser and such that, if $D \retract D \to TD$ then
$TD$ is a model of the conversion relation induced by the reduction, and we call it a $T$-model.

Intersection types are naturally interpreted as predicates over a $\lambda$-model, and indeed intersection type systems have been originally conceived to characterize
strongly normalizing, weakly normalizing and solvable terms namely having head normal form. 
Toward developing analogous systems for the  
computational $\lambda$-calculus, we introduce an intersection type assignment system
with two sorts of intersection types, namely {\em value types} ranged over by $\delta$, and {\em computation types} ranged over by $\tau$, whose intended meanings are subsets of $D$ and $TD$, respectively. We then define the minimal type theories $\Th_\Val$ and $\Th_\Comp$ axiomatizing  
the preorders over value and computation types respectively, and construct a type assignment system which is a generalization of the BCD type 
system for the ordinary $\lambda$-calculus in \cite{BCD'83}. Then, the subject reduction property smoothly follows, and can be established along the lines of the analogous property of system BCD. 

We are looking at BCD type system because it defines a logical semantics of $\lambda$-terms, whose meaning are just the sets of types that can be assigned to them, which turn out to be filters of types. 
Such a model, named {\em filter model}, has the structure of an algebraic lattice with countable basis. 
This fact is at the hearth of the proof
of completeness of the system, namely  that the denotation of a term belongs to the interpretation of a type in any model if and only if
the type can be assigned to the term in the type system.

However, the type interpretation over a $T$-model is much more problematic than in case of intersection types and $\lambda$-models. The issue is ensuring that computation types
are closed under the two basic operations of the monad $T$, that is unit and bind, which we dub {\em monadic type interpretations}.
As we shall see in the technical development, the natural clauses
lead to a not inductive definition, hence not inductive. To solve the problem we cannot resort to the correspondence of intersection types to compact
points in $D$ and $TD$, because there is no information about the compacts of $TD$, 
since the monad $T$ is a parameter.

The solution we propose is to restrict type interpretation to the case of $T$-models that are (pre-)fixed points of the functor $\textbf{F}(X) = (X \to TX)$, existing as inverse limit constructions if $T$ and therefore $\textbf{F}$ are $\omega$-continuous functors. What one obtains in this way is an instance of Scott's $D_\infty$
model, which is the co-limit of a chain of approximant domains $D_n$. By interpreting types as admissible subsets of the $D_n$ by induction over $n$, we obtain admissible subsets of $D_\infty$ and $TD_\infty$
by the very same co-limit construction.

Coming to the filter model construction, we build over the fact that such models can be seen as inverse limit domains, whose structure is determined by the type preorder, that is the type theory one considers: see in particular \cite{Dezani-CiancagliniHA03} and \cite{BarendregtDS2013} 16.3. To avoid the rather 
inelegant shape of domain equations arising 
from non-extensional filter models, we show here how an extensional $T$-model can be constructed
as a filter model, that is itself a limit model satisfying the domain equation $D = D \to TD$. This eventually leads to the completeness theorem, of which subject expansion is a corollary. We then define a natural convergence predicate, and characterize convergent terms via their non trivial typings, and conclude that the filter model is
computationally adequate.

We assume familiarity with $\lambda$-calculus, intersection types and domain theory, for which we refer to
textbooks such as \cite{Amadio-Curien'98} and \cite{BarendregtDS2013} part III.
Due to space restrictions most proofs are either sketched or omitted, or moved to the Appendix. 

\section{The untyped computational $\lambda$-calculus: $\TFlambdaComp$}\label{sec:calculus}
The syntax of the untyped computational $\lambda$-calculus differs from Moggi's original definition
of monadic metalanguages, and it is closer to Wadler's in \cite{Wadler-Monads}.
As said in the Introduction the untyped computational calculus $\TFlambdaComp$ has two kinds of terms.

\begin{definition}[Values and computations]\label{def:terms}	
The {\em untyped computational $\lambda$-calculus}, shortly $\TFlambdaComp$, is a calculus of two sorts of expressions:
\[\begin{array}{r@{\hspace{0.7cm}}rll@{\hspace{1cm}}l}
\ValTerm: & V, W & ::= & x \mid \lambda x.M & \mbox{(values)} \\ [1mm]
\ComTerm: & M,N & ::= & \Unit V \mid M \Bind V & \mbox{(computations)}
\end{array}\]
where $x$ ranges over a denumerable set $\Var$ of variables. We set $\Term = \ValTerm\, \cup \ComTerm$; 
$\FV(V)$ and $\FV(M)$ are the sets of free variables occurring in $V$ and $M$ respectively, and are defined in the obvious way.

\end{definition}

\begin{remark}\label{rem:terms}
In the above definition the sorts correspond to the two types $D$ and $TD$. 
By considering $D \retract D \to TD$ and setting $\Unit: D \to TD$ and $\Bind: TD \times D \to TD$, 
we see that all $V\in \ValTerm$ are of type $D$, and
all $M \in \ComTerm$ are of type $TD$.

The reduction rules in the next subsection are inspired to \cite{Wadler-Monads} and clearly reflect
the equations in Definition \ref{def:monad}. However, Wadler
defines an equational theory instead of a reduction relation, and his definition has a much richer type system, since 
the calculus is simply typed and there are types like $TTA$ or $T(A\to B)$ etc. 
\end{remark}

\subsection{Reduction}\label{subsec:reduction}


Following \cite{Barendregt'84} \S 3.1, we say that a binary relation $R \subseteq \Term \times \Term$ is a 
{\em notion of reduction}. If $R_1,R_2$ are notions of reductions we abbreviate
$R_1 R_2 = R_1 \cup R_2$; then
we denote by $\Red_R$ the {\em compatible closure} of $R$, namely the least relation including $R$ which is
closed under arbitrary contexts.

\begin{definition}[Reduction]\label{def:reduction}
	The  relation $\MonRed = \Betac \cup \BindRight \cup \BindLeft$ is the union of the following binary relations over $\ComTerm$:
	\[
	\begin{array}{rcl}
	\Betac & = &\Set{\Pair{\Unit V \Bind (\lambda x.M)}{ M\Subst{V}{x}} \mid V \in \ValTerm, M\in \ComTerm} \\ [1mm]
 	\BindRight & = &  \Set{ \Pair{M \Bind \lambda x. \Unit x}  {M} \mid M \in \ComTerm} \\ [1mm]
 	\BindLeft &  = & \Set{\Pair{(L \Bind \lambda x.M) \Bind \lambda y.N}{L  \Bind \lambda x. (M \Bind \lambda y.N)} \mid L,M,N \in \ComTerm, x\not\in FV(N)}
	\end{array}
	\]
	where $M\Subst{V}{x}$ denotes the capture avoiding substitution of $V$ for $x$ in $M$. 
	
	\medskip \noindent
	Finally $\Red\; =\; \Red_{\MonRed}$ is the compatible closure of $\MonRed$.
\end{definition}
A more readable writing of the definition of $\Red$ in Def. \ref{def:reduction} is:
\[
	\begin{array}{r@{\hspace{0.4cm}}rll@{\hspace{0.4cm}}l}
	\Betac) & \Unit V \Bind (\lambda x.M) & \Red & M\Subst{V}{x} \\ [1mm]
	\BindRight \,) &  M \Bind \lambda x. \Unit x & \Red & M \\ [1mm]
	\BindLeft\,) &  (L \Bind \lambda x.M) \Bind \lambda y.N & \Red & L  \Bind \lambda x. (M \Bind \lambda y.N)  & \mbox{for $x\not\in FV(N)$}
	\end{array}
\]

Rule $\Betac$ is reminiscent of the left unit law in \cite{Wadler-Monads}; we call it $\Betac$ because it performs call-by-value 
$\beta$-contraction in $\TFlambdaComp$. 
In fact, by reading $\Bind$ as postfix functional application and merging $V$ into its trivial computation
$\Unit V$, $\Betac$ is the same as $\beta_v$ in \cite{Plotkin'75}:

\begin{equation}\label{eq:betav}
(\lambda x.M)V \equiv \Unit V \Bind (\lambda x.M) \Red M\Subst{V}{x} 
\end{equation}

The compatible closure of the relation $\Betac \cup \BindRight \cup \BindLeft$ 
is explicitly defined by means of the typed contexts:
\[\begin{array}{l@{\hspace{0.5cm}}lcl}
\textit{Value contexts:} & \ValCtx & ::= & \Hole{\cdot}{D}  \mid \lambda x.\, \CompCtx \\[1mm]
\textit{Computation contexts:} & \CompCtx & ::= &  \Hole{\cdot}{TD} \mid \Unit \ValCtx \mid \CompCtx \Bind V \mid M\Bind \ValCtx
\end{array}\]
Contexts have just one hole, which is either $\Hole{\cdot}{D}$ or $\Hole{\cdot}{TD}$. These are typed in the sense that
they can be replaced only by value and computation terms respectively. Denoting
by $\ValCtx\Hole{P}{}$ and $\CompCtx\Hole{P}{}$ the replacements 
of hole $\Hole{\cdot}{D}$ or $\Hole{\cdot}{TD}$ in $\ValCtx$ and $\CompCtx$ by $P\equiv V$ or $P\equiv M$ 
respectively (possibly catching free variables in $P$) we
get terms in $\ValTerm$ and $\ComTerm$ respectively. Then compatible closure is expressed by the rule:

\begin{equation}\label{eq:compatibleClosure}
\prooftree
M \Red M'
\justifies
\CompCtx\Hole{M}{} \Red \CompCtx\Hole{M'}{}
\endprooftree
\end{equation}

\medskip
Using rule (\ref{eq:compatibleClosure})
the correspondence of rule $\Betac$ to $\beta_v$, can now be illustrated more precisely. First observe that the reduction relation 
$\Red$ is only among computations, therefore
no computation $N$ will ever reduce to some value $V$; however this is represented by a reduction $N \RedStar \Unit V$, where $ \Unit V$ is
the coercion of the value $V$ into a computation. Moreover, let us assume that $M \RedStar \Unit (\lambda x.M')$; by setting
\begin{equation}\label{eq:monadicApp}
MN \equiv M \Bind( \lambda z. N \Bind z) \qquad \mbox{for $\quad z \not\in FV(N)$}
\end{equation}
we have:
\[\begin{array}{lll@{\hspace{0.5 cm}}l}
MN & \RedStar & \Unit (\lambda x.M') \Bind( \lambda z. \Unit V \Bind z) & \mbox{by rule  (\ref{eq:compatibleClosure})} \\ [1mm]
& \Red & \Unit V \Bind  (\lambda x.M') & \mbox{by $\Betac$} \\ [1mm]
& \Red & M'\Subst{V}{x} & \mbox{by $\Betac$}
\end{array}\]
where if $z \not\in FV(N)$ then $z \not\in FV(V)$ as it can be shown by a routine argument.

\medskip
We end this section by considering the issue of weak and full extensionality, that have not been treated in Def. \ref{def:reduction}.
Weak extensionality, also called $\xi$-rule of the ordinary $\lambda$-calculus, is reduction under abstraction. This is guaranteed by
rule (\ref{eq:compatibleClosure}), but only in the context of computation terms. 

Concerning extensionality an analogous of $\eta$-rule is:
\begin{equation}
\eta_c) \quad \lambda x. \, (\Unit \, x \Bind \, V) \Red_{\eta_c} V, \qquad x \not \in FV(V)
\end{equation}
This involves extending reduction from $\ComTerm$ to the whole $\Term$. However the reduction obtained by
adding $\eta_c$ to $\MonRed$ is not confluent:
\begin{diagram}[width=4em]
		(M \Bind \lambda x.\,(\Unit x \Bind y)) \Bind z &	&  \rTo^{\BindLeft}	&  &	M \Bind \lambda x\,((\Unit x \Bind y) \Bind z)  \\
		\dTo^{\eta_c} & & & &  \dDashto \\
		(M \Bind y) \Bind z & & \rDashto  & & ? 
\end{diagram}

\subsection{Confluence}\label{subsec:confluence}

A fundamental property of reduction in ordinary $\lambda$-calculus is confluence, established in the Church-Rosser theorem.
In this section we prove confluence of $\Red$ for the $\TFlambdaComp$-calculus. This is a harder task since
reduction in $\TFlambdaComp$ has three axioms instead of just the $\beta$-rule of the $\lambda$-calculus, 
whose left-hand sides generate a number of critical pairs. Before embarking into the proof let us see a few examples.

\begin{example}\label{ex:confluence-a}
	In this example we see how outer reduction by $\BindLeft$ may overlap by an inner reduction by $\Betac$. Representing the given reductions by solid arrows,
	we see how to recover confluence by a reduction and a relation represented by a dashed arrow and a dashed line, respectively:
	\begin{diagram}[width=4em]
		(\Unit V \Bind \lambda x\,M) \Bind \lambda y.\, N  &	&  \rTo^{\BindLeft}	&  &	\Unit V \Bind \lambda x\,(M \Bind \lambda y.\, N)  \\
		\dTo^{\Betac} & & & & \dDashto_{\Betac} \\
		M \Subst{V}{x} \Bind \lambda y.\, N & & \rDashes_{\equiv} & & (M \Bind \lambda y.\, N) \Subst{V}{x} 
	\end{diagram}
	\noindent where $x\not \in \FV(N)$, which is the side condition to rule $\BindLeft$; therefore  the two terms in the lower line of the diagram
	are syntactically identical. 
\end{example}

\begin{example}\label{ex:confluence-b}
	In this example we see how outer reduction by $\BindLeft$, overlapping with outer $\BindRight$, can be recovered by an inner reduction by $\BindRight$:
	\begin{diagram}[width=4em]
		(M \Bind \lambda y\,N) \Bind\lambda x.\, \Unit x   &	&  \rTo^{\BindLeft}	&  &	M \Bind \lambda y.\, (N \Bind  \lambda x.\, \Unit x)  \\
		& \rdTo_{\BindRight}		&   									&     \ldDashto_{\BindRight} 	&	\\
		&& M \Bind \lambda y.\, N	&&  \\
	\end{diagram}
	
\end{example}

\begin{example}\label{ex:confluence-c}
	Here the outer  reduction by $\BindLeft$ overlaps with an inner reduction by $\BindRight$. This is recovered by means of an inner reduction by $\Betac$:
	\begin{diagram}[width=4em]
		(M \Bind \lambda x.\, \Unit x) \Bind \lambda y.\, N & &	 \rTo^{\BindLeft}	& & 	M \Bind \lambda x.\, (\Unit x \Bind \lambda y\, N)  \\
		\dTo^{\BindRight}		&   	&							&	&     \dDashto_{\Betac} 		\\
		M \Bind \lambda y.\, N	&	&        \rDashes_{\alpha}		&	& M \Bind \lambda x.\, N\Subst{x}{y}	  \\
	\end{diagram}
	
	\noindent where $x\not \in \FV(N)$ as observed in Example \ref{ex:confluence-a}, and therefore $ \lambda x.\, N\Subst{x}{y}$ is the renaming by $x$ of the bound variable
	$y$ in $\lambda y.\, N$: then the dashed line represents $\alpha$-congruence.
	
\end{example}

After having inspected the above examples, one might by tempted to conclude that the reduction in Definition \ref{def:reduction} enjoys the diamond property, namely it is
confluent within (at most) two single steps, one per side (for the diamond property see (\ref{eq:DP}) below: we say here `at most' because $\Red$ is not reflexive). Unfortunately, this is not the case because of rule $\Betac$, that can multiplicate redexes in the reduced term exactly as  the $\beta$-rule in ordinary $\lambda$-calculus.
Even worse, rule $\BindLeft$ generates critical pairs with all other rules and with itself, preventing the simple extension of standard proof methods to succeed. 
For these reasons, we to split the proof in several steps, proving confluence of $\Betac \cup \BindRight$ and $\BindLeft$ separetedly, and then combining
these results by means of the commutativity of these relations and Hindley-Rosen lemma\footnote{We are grateful to ... for suggesting this proof strategy.}.

\medskip

In the first step we adapt a method originally due to Tait and Martin L\"of, 
and further developed by Takahashi \cite{takahashi95:_paral}. See e.g. the book \cite{Terese03} ch. 10. 
Let's define the following relation $\circlearrow$:

\begin{definition}\label{def:appSimReddef}
	The relation $\circlearrow\; \subseteq \Term\times\Term$ is inductively defined by:
	\begin{enumerate}[label=\roman*)]
		\item  \label{def:appSimReddef-i} $x \circlearrow x$
		\item  \label{def:appSimReddef-ii} $M\circlearrow N\Then \lambda x.M\circlearrow \lambda x.N$
		\item  \label{def:appSimReddef-iii} $V\circlearrow V' \Then \Unit V\circlearrow \Unit V'$
		\item  \label{def:appSimReddef-iv} $M\circlearrow M' \mbox{ and } V\circlearrow V' \Then M\Bind V\circlearrow M'\Bind V'$
		\item  \label{def:appSimReddef-v} $M\circlearrow M' \mbox{ and } V\circlearrow V' \Then \Unit V\Bind \lambda x.M \circlearrow M'\Subst{V'}{x}$
		\item \label{def:appSimReddef-vi}  $M\circlearrow M' \Then M\Bind \lambda x.\Unit x \circlearrow M'$
	\end{enumerate}	
\end{definition}	

By (\ref{def:appSimReddef-i} - (\ref{def:appSimReddef-iv} above, relation $\circlearrow$ is reflexive and coincides with its compatible closure.
Also $\BetaIdRed \subseteq\; \circlearrow$; intentionally, this is not the case w.r.t. the whole $\Red$.

\begin{lemma}\label{lem:SubLemmaSimRed}
	For $M, M'\in \ComTerm$ and $V, V'\in\ValTerm$ and every variable $x$, if $M\circlearrow M'$ and $V\circlearrow V'$, 
	then $M\Subst{V}{x} \circlearrow M'\Subst{V'}{x}$.
\end{lemma}

\begin{proof}
	By an easy induction on the definition of $M\circlearrow M'$ and $V\circlearrow V'$.
\end{proof}

\noindent Now, by means of Lemma \ref{lem:SubLemmaSimRed} one easily proves that $\circlearrow \; \subseteq \; \BetaIdRedStar$.

\medskip
The next step in the proof is to show that the relation $\circlearrow$ satisfies the {\em triangle property} \TP:
\begin{equation}\label{eq:TP}  
\forall P \, \exists P^* \, \forall Q. ~ P \circlearrow Q \; \Then \; Q \circlearrow P^* 
\end{equation}
where $P, P^*, Q \in \Terms$. \TP\, implies the {\em diamond property} \DP, which for $\circlearrow$ is:
\begin{equation}\label{eq:DP}  
\forall P,Q,R. \; P \circlearrow Q \And P \circlearrow R \; \Then \; \exists P'.\; Q \circlearrow P' \And R \circlearrow P'
\end{equation}
In fact, if \TP \, holds then we can take $P' \equiv P^*$ in \DP, since the latter only depends on $P$. We then define $P^*$ in terms of $P$ as follows:
\begin{enumerate}[label=\roman*)]
	\item $x^*\equiv x$
	\item $(\lambda x. M)^* \equiv \lambda x.M^*$
	\item $(\Unit V)^* \equiv \Unit V^*$ 
	\item $(\Unit V \Bind \lambda x.M)^* \equiv M^*\Subst{V^*}{x}$	
	\item $(M \Bind \lambda x.\Unit x)^* \equiv M^*$, if $M \not\equiv \Unit V$ for $V\in \ValTerm$
	\item $(M \Bind V)^* \equiv M^* \Bind V^*$, $M \not\equiv \Unit W$ for $W\in \ValTerm$ and $V \not\equiv \lambda x.\Unit x$
\end{enumerate}

\begin{lemma}\label{lem:TPSimRed}
	For all $P,Q\in \Term$, if $P\circlearrow Q$ then $Q\circlearrow P^*$, namely $\circlearrow$ satisfies \TP.
\end{lemma}

\begin{proof} By induction on $P\circlearrow Q$. The base case $x \circlearrow x$ follows by $x^* \equiv x$. All remaining cases
follow by the induction hypotheses; in particular if $P \equiv \Unit V\Bind \lambda x.M \circlearrow M'\Subst{V'}{x} \equiv Q$
because $M\circlearrow M'$ and $V\circlearrow V'$, then by induction $M' \circlearrow M^*$ and $V' \circlearrow V^*$, so that
$M'\Subst{V'}{x} \circlearrow M^*\Subst{V^*}{x} \equiv P^*$ by Lem. \ref{lem:SubLemmaSimRed}.
\end{proof}

A notion of reduction $R$ is said to be {\em confluent} or {\em Church-Rosser}, shortly $\CR$, if $\RedStar_R$ satisfies $\DP$; more explicitly
for all $M, N, L \in \ComTerm$:
\[M \RedStar_R N \And M \RedStar_R L \Then \exists M' \in \ComTerm.\, N \RedStar_R M' \And L \RedStar_R M'\]

\begin{corollary}\label{cor:confluence-bcid}
	The notion of redution $\Betac \cup \BindRight$ is $\CR$.
\end{corollary}

\begin{proof}
	As observed above $\BetaIdRed \subseteq\; \circlearrow$, hence $M \BetaIdRedStar N$ implies
	$M \circlearrow^+ N$, where $\circlearrow^+$ is the transitive closure of $\circlearrow$, and similarly
	$M \circlearrow^+ L$. By Lemma \ref{lem:TPSimRed} $\circlearrow$ satisfies \TP, hence it
	satisfies $\DP$. By an easy argument (see e.g. \cite{Barendregt'84} Lemma 3.2.2) we conclude that 
	$N \circlearrow^+ M'$ and $L \circlearrow^+  M'$ for some $M'$, from which the thesis follows by the fact that  
	$\circlearrow \; \subseteq \;  \BetaIdRedStar$.
\end{proof}

A notion of reduction $R$ is  {\em weakly Church-Rosser}, shortly $\WCR$, if 
for all $M, N, L \in \ComTerm$:
\[M \Red_R N \And M \Red_R L \Then \exists M' \in \ComTerm.\, N \RedStar_R M' \And L \RedStar_R M'\]

\begin{lemma}\label{lem:ass-wcr}
The notion of reduction $\BindLeft$ is $\WCR$.
\end{lemma}

\begin{proof} It suffices to show the thesis for the critical pair $M_1 \AssRed M_2$ and $M_1 \AssRed M_3$ where:
\[\begin{array}{rcl}
M_1 & \equiv & ((L \Bind \lambda x.M) \Bind \lambda y,N) \Bind \lambda z.P \\
M_2 & \equiv & (L \Bind \lambda x.M) \Bind \lambda y.(N \Bind \lambda z.P) \\
M_3 & \equiv & (L \Bind \lambda x.(M \Bind \lambda y.N)) \Bind \lambda z.P
\end{array}\]
Then in one step we have:
\[M_2 \AssRed  L \Bind \lambda x.(M \Bind \lambda y.(N \Bind \lambda z.P)) \equiv M_4\]
but 
\[M_3 \AssRed L \Bind \lambda x.((M \Bind \lambda y.N) \Bind \lambda z.P) \AssRed M_4\]
where the two reduction steps are necessary.
\end{proof}

A notion of reduction $R$ is {\em noetherian} or {\em strongly normalizing}, shortly $\SN$, if there exists
no infinite reduction $M \Red_R M_1 \Red_R M_2 \Red_R \cdots$ out of any $M \in \ComTerm$.

\begin{lemma}\label{lem:ass-sn}
The notion of reduction $\BindLeft$ is $\SN$.
\end{lemma}

\begin{proof} Given $M \in \ComTerm$ let's denote by the same $M$ the expression obtained by marking
differently all occurrences of $\Bind$ in $M$, say $\Bind_1, \ldots, \Bind_n$. We say that $\Bind_i$ is to the left to
$\Bind_j$ in $M$ if there exists a subterm $L \Bind_j V$ of $M$ such that $\Bind_i$ occurs in $L$. Finally
let's denote by $\sharp M$ the number of pairs $(\Bind_i,\Bind_j)$ such that $\Bind_i$ is to the left to
$\Bind_j$ in $M$.

If a term includes an $\BindLeft$-redex $(L \Bind_i \lambda x.N) \Bind_j \lambda y.P$, which is contracted 
to $L \Bind_i \lambda x.(N \Bind_j \lambda y.P)$, then $\Bind_i$ is to the left to $\Bind_j$ in the redex,
but not in the contractum. Also it is easily seen by induction on terms that, if $\Bind_i$ is not to the left to
$\Bind_j$ in $M$ and $M \AssRed N$, the same holds in $N$.

It follows that, if $M \AssRed N$ then $\sharp M > \sharp N$, hence  $\BindLeft$ is $\SN$.
\end{proof}

\begin{corollary}\label{cor:ass-cr}
The notion of reduction $\BindLeft$ is $\CR$.
\end{corollary}

\begin{proof} By Lem. \ref{lem:ass-wcr}, \ref{lem:ass-sn} and by Newman Lemma (see \cite{Barendregt'84}, Prop. 3.1.24), stating that
a notion of reduction which is $\WCR$ and $\SN$ is $\CR$.
\end{proof}

\newcommand{\LRed}[1]{_{#1}\!\!\longleftarrow}
\newcommand{\LRedStar}[1]{_{#1}\!\!\stackrel{*}{\longleftarrow}}
\newcommand{\RedEq}{\stackrel{=}{\Red}}

The following definitions are from \cite{BaaderN98}, Def. 2.7.9.
Two relations $\Red_1$ and $\Red_2$ over $\ComTerm$ are said to {\em commute} if, for all $M,N,L$:
\[N\, \LRedStar{1} M \RedStar_2 L \Then\exists P \in \ComTerm.\; N \RedStar_2 P\, \LRedStar{1} L \]
Relations $\Red_1$ and $\Red_2$ {\em strongly commute} if, for all $M,N,L$:
\[N\, \LRed{1} M \Red_2 L \Then\exists P \in \ComTerm.\; N \RedEq_2 P\, \LRedStar{1} L \]
where $\RedEq_2 $ is $\Red_2 \cup =$, namely at most one reduction step.

\begin{lemma}\label{lem:commutativity}
Reductions $\BetaIdRed$ and $\AssRed$ commute.
\end{lemma}

\begin{proof} By Lemma 2.7.11 in \cite{BaaderN98}, two strongly commuting relations commute, and
commutativity is clearly symmetric; hence it suffices to show that
\[N\, \LRed{\beta_c,\textit{id}} M \AssRed L \Then\exists P \in \ComTerm.\; N \RedEq_{\textit{ass}} P\, \LRedStar{1} L. \]
We can limit the cases to the critical pairs, that are exactly those in examples \ref{ex:confluence-a}, \ref{ex:confluence-b} and
\ref{ex:confluence-c}, which commute.
\end{proof}

\begin{theorem}[Confluence]\label{thm:confluence}
The notion of reduction $\MonRed = \Betac \cup \BindRight\, \cup \BindLeft$ is $\CR$.
\end{theorem}

\begin{proof} By the commutative union lemma (see \cite{BaaderN98}, Lem. 2.7.10 and \cite{Barendregt'84}, Prop. 3.3.5, 
where it is called Hindley-Rosen Lemma),
if $\BetaIdRed$ and $\AssRed$ and are both $\CR$ (Cor. \ref{cor:confluence-bcid} and \ref{cor:ass-cr}), and commute
(Lem. \ref{lem:commutativity}), then $\Red_{\MonRed} \;=\; \BetaIdRed \cup \AssRed$ is $\CR$.
\end{proof}

\section{Models of $\TFlambdaComp$}\label{subsec:models}

Let $\Dcat$ be a category of domains, namely a cartesian closed subcategory of the category of posets 
whose objects have directed sups and morphisms are Scott continuous functions. Below $|\Dcat|$ denotes the set of objects of $\Dcat$.


\begin{definition}[Monad]\label{def:monad}
A {\em monad} over $\Dcat$ is a triple $(T, \Unit\!\!, \,\Bind)$ where $T:|\Dcat| \to |\Dcat|$ is a map over the objects of $\Dcat$,  
$\Unit \!\!= \Set{\UnitSub{D} \mid D \in |\Dcat|}$ and $\Bind = \Set{\Bind_{D,E} \mid D,E \in |\Dcat|}$ are families of morphisms
\[\UnitSub{D}: D \to TD, \qquad \Bind_{D,E}: TD \times (D \to TE) \to TE\]
such that, writing functional application $f(x)$ as $f\,x$, $\Bind$ as infix operator and omitting subscripts:
\[(\Unit\,d) \Bind f = f\,d \qquad a \Bind \Unit = a \qquad (a \Bind f) \Bind g = a \Bind \metalambda d. (f\,d \Bind g) \]
\end{definition}

This definition of a monad, akin to that of a Kleisli triple, is the type theoretic definition by Wadler in \cite{Wadler-Monads}, at the basis of
Haskell implementation of monads. We use this instead of the category theoretic definition, originally used by Moggi in \cite{Moggi'91} as it is more accessible to non categorist readers. 
If $T$ is a monad over $\Dcat$, we say that a {\em $T$-model} of $\TFlambdaComp$ is a call-by-value reflexive object in
$\Dcat$ (see \cite{Moggi'89}).

\begin{definition}[$T$-model]\label{def:T-model}
A {\em $T$-model} in the category $\Dcat$ is a tuple $(D, T, \Phi, \Psi)$ where $D\in |\Dcat|$,  
$(T, \Unit\!\!, \,\Bind)$ is a monad over $\Dcat$, and $\Phi: D \to (D \to TD)$ and $\Psi: (D\to TD) \to D$ are 
morphisms in $\Dcat$ such that $\Phi \circ \Psi = \Id_{D\to TD}$.
A $T$-model  is {\em extensional} if also $\Psi \circ \Phi = \Id_D$, namely $D \simeq D \to TD$ in $\Dcat$.
\end{definition}
In the following we just say that some $D$ is a $T$-model, when the  monad $(T, \Unit\!\!, \,\Bind)$ and the injection-projection
pair $(\Phi,\Psi)$ are understood.

\begin{remark}\label{rem:T-model}
The definition of $T$-model is the call-by-value generalization of that of $\lambda$-model, where $D$ is a retract of $D\to D$. Also a
call-by-name notion of model is possible by considering a retract of $TD \to TD$ instead. We concentrate on call-by-value $\TFlambdaComp$
as it is a more natural model of  effectful functional calculi.
\end{remark}

In case of a $T$-model we are interested to $\UnitSub{D}$ and $\Bind_{D,D}$, which are respectively the intended meanings of $\Unit$ and $\Bind$ operators 
in the computation syntax. We deliberately overload notations and avoid subscripts when unnecessary.

\begin{definition}\label{def:interpretation}
Let $D$ be a $T$-model and $\Env_D = \Var \to  D$ be the set of {\em variable interpretations} into $D$
ranged over by $\rho$, then the maps 
\[\Sem{\cdot}^D : \ValTerm \to \Env_D \to D,  \qquad \Sem{\cdot}^{TD} : \ComTerm \to \Env_D \to TD\] 
are defined by mutual induction:
	\[\begin{array}{rcl@{\hspace{1cm}}rcl}
	\Sem{x}^D_\rho & = & \rho(x)  & \Sem{\Unit V}^{TD}_\rho & = & \Unit \Sem{V}^D_\rho \\ [2mm]
	\Sem{\lambda x. M}^D_\rho & = & \metalambda d \in D. \, \Sem{M}^{TD}_{\rho[x \mapsto d]} 
	& \Sem{M \Bind V}^{TD}_\rho & = & \Sem{M}^{TD}_\rho \,\Bind\, \Sem{V}^D_\rho 
	\end{array}\]
where $\rho[x \mapsto d](y) = \rho(y)$ if $y \not \equiv x$, it is equal to $d$ otherwise.
\end{definition}

\begin{lemma}\label{lem:subtsInterpretation}
In any $T$-model $D$ we have $\Sem{W\Subst{V}{x}}^{D}_\rho = \Sem{W}^{D}_{\rho[x\mapsto \Sem{V}^D_\rho]}$
and 
$\Sem{M\Subst{V}{x}}^{TD}_\rho = \Sem{M}^{TD}_{\rho[x\mapsto \Sem{V}^D_\rho]}$.
\end{lemma}

\begin{proposition}\label{prop:soundness-of-interpretation}
If $M \Red N$ then $\Sem{M}^{TD}_\rho = \Sem{N}^{TD}_\rho$ for any $T$-model $D$ and $\rho \in \Env_D$. Therefore,
if $=$ is the convertibility relation of $\Red$, that is the symmetric closure of $\RedStar$, then $M = N$
implies $\Sem{M}^{TD}_\rho = \Sem{N}^{TD}_\rho$.
\end{proposition}

\section{Intersection type assignment system for $\TFlambdaComp$}\label{sec:types}

\begin{definition}[Intersection types and Type theories]\label{def:theories}
A  {\em language of intersection types} $\Types$ is a set of expressions $\sigma, \sigma', \ldots$ including a constant $\omega$ 
and closed under the intersection operator: $\sigma\Inter\sigma'$.

An {\em intersection type theory} (shortly a {\em type theory}) is a pair $\Th = (\Types, \leq)$ where $\Types$ is a language of intersection types and $\leq$ a pre-order over $\Types$ such that $\omega$ is the top, 
$\Inter$ is monotonic, idempotent and commutative, and
\[\sigma \Inter \sigma' \leq \sigma, \qquad 
\prooftree
	\sigma \leq \sigma' \quad \sigma \leq \sigma''
\justifies
	\sigma \leq \sigma'\Inter\sigma''
\endprooftree
\]
\end{definition}

\begin{definition}[Intersection types for values and computations]\label{def:intersectionTypesValueComp}

Let $\TypeVar$ be a countable set of {\em type variables}, ranged over by $\alpha$:
\[\begin{array}{r@{\hspace{0.7cm}}rll@{\hspace{0.7cm}}l}
\ValType: & \delta & ::= & \alpha \mid \delta  \rightarrow\tau \mid \delta \wedge \delta \mid \omega_\Val  & \mbox{({\em value types})}\\ [1mm]
\ComType: & \tau & ::= & T\delta \mid \tau \wedge\tau \mid \omega_\Comp & \mbox{({\em computation types})}
\end{array}\]
\end{definition}
Intersection types are better understood as predicates of values and computations respectively, or as refinement types of the two types
of $\TFlambdaComp$, that is, using the notation in 
\cite{MelliesZeilberger2015}, $\delta \sqsubset D = D\to TD$ in case of values, and $\tau \sqsubset TD$ in case of computations.

In the definition of language $\ValType$ and consequently $\ComType$ the set of $\TypeVar$ (also called {\em atoms}) is left unspecified and it is a parameter

\begin{definition}[Type theories $\Th_\Val$ and $\Th_\Comp$]\label{def:type-theories-Th_V-Th_C}
The intersection type theories $\Th_\Val = (\ValType, \leq_\Val)$ and $\Th_\Comp = (\ComType, \leq_\Comp)$ are the least type theories such that:
\[ \begin{array}{c@{\hspace{1cm}}c}
\delta \leq_\Val \omega_\Val & \omega_\Val \leq_\Val \omega_\Val \to \omega_\Comp 
\end{array}\]
\[ \begin{array}{c@{\hspace{1cm}}c}
(\delta \to \tau) \Inter (\delta \to \tau') \leq_\Val \delta \to (\tau \Inter \tau') & 
\prooftree
	\delta' \leq_\Val \delta \quad \tau \leq_\Comp \tau'
\justifies
	\delta \to \tau \leq_\Val \delta' \to \tau'
\endprooftree \\ [1mm]
\end{array}\]
\[ 
\begin{array}{c@{\hspace{1cm}}c@{\hspace{1cm}}c}
\tau \leq_\Comp \omega_\Comp & 
T\delta \Inter T\delta' \leq_\Comp T(\delta \Inter \delta') & 
\prooftree
	\delta \leq_\Val \delta'
\justifies
	T\delta \leq_\Comp T\delta'
\endprooftree\\ 
\end{array}
\]
\end{definition}

\begin{remark}\label{rem:type-theories-Th_V-Th_C}
Writing $=_\Val$ and $=_\Comp$ for the antisymmetric closure of $\leq_\Val$ and $\leq_\Comp$ respectively, we see that all the axioms but $\delta \leq_\Val \omega_\Val$ and $\tau \leq_\Comp \omega_\Comp$ are actually equalities. 
\end{remark}

\begin{lemma}\label{lem:tau-neq-omega}
If $\tau \in \ComType$ is such that $\tau \neq_\Comp \omega_\Comp$ then for some $\delta \in \ValType$ we have $\tau =_\Comp T\delta$; hence $\tau \leq_\Comp T\omega_\Val$.
\end{lemma}

\Proof
By induction over $\tau$. The only non trivial case is when $\tau \equiv \tau_1 \Inter \tau_2$. From $\tau_1 \Inter \tau_2  \neq_\Comp \omega_\Comp$ 
it follows that at least one of them is different than $\omega_\Comp$: if say $\tau_1 =_\Comp \omega_\Comp$ then $\tau_1 \Inter \tau_2 = _\Comp \tau_2 \neq_\Comp \omega_\Comp$ os that $\tau_2 =_\Comp T\delta_2$ by induction. Finally if both $\tau_1$ and $\tau_2$ are not equated to $\omega_\Comp$ then by induction
$\tau_1 \Inter \tau_2 = _\Comp T\delta_1 \Inter T\delta_2 =_\Comp T(\delta_1 \Inter \delta_2) \leq_\Comp T\omega_\Val$, for some $\delta_1,\delta_2 \in \ValType$.
\QED

\begin{definition}[Type assignment]\label{def:typeAssignment}
A {\em basis} is a finite set of typings $\Gamma = \Set{x_1:\delta_1, \ldots x_n: \delta_n}$ with pairwise distinct variables $x_i$, whose {\em domain} is the set 
$\Dom(\Gamma) = \Set{x_1, \ldots, x_n}$. A basis determines a function from variables to types such that $\Gamma(x) = \delta$ if $x:\delta \in \Gamma$,
$\Gamma(x) = \omega_\Val$ otherwise.

A {\em judgment} is
an expression of either shapes: $\Gamma \der V:\delta$ or $\Gamma \der M:\tau$. It is {\em derivable} if it is the conclusion of a derivation 
according to the rules:
\[ \arraycolsep=5pt\def\arraystretch{2.8}
\begin{array}{crcr}
\prooftree
x:\delta \in \Gamma
\justifies
\Gamma \der x:\delta
\using  \mbox{(Ax)}
\endprooftree 
& &
\prooftree
\Gamma, x:\delta \der M:\tau
\justifies
\Gamma \der \lambda x.M: \delta \to \tau
\using  \IntrArr
\endprooftree
& \\
\prooftree
\Gamma \der V:\delta
\justifies
\Gamma \der \Unit V: T\delta 
\using \IntrUnit 
\endprooftree
&  &
\prooftree
\Gamma \der M: T\delta \quad \Gamma \der V:\delta\to \tau
\justifies
\Gamma \der M \Bind V: \tau 
\using \ElimArr
\endprooftree & 
\end{array}
\]
where $\Gamma, x:\delta = \Gamma \cup \Set{x:\delta}$ with $x:\delta\not\in\Gamma$, and the rules:
\[\begin{array}{c@{\hspace{0.4cm}}c@{\hspace{0.4cm}}c}
\prooftree
\vspace{0.3cm}
\justifies
\Gamma \der P:\
\endprooftree \quad (\omega)
&
\prooftree
\Gamma \der P:\sigma \quad \Gamma \der P:\sigma'
\justifies
\Gamma \der P:\sigma \Inter \sigma'
\endprooftree \quad \IntrInter
&
\prooftree
\Gamma \der P:\sigma \quad \sigma \leq \sigma'
\justifies
\Gamma \der P:\sigma' 
\endprooftree \quad (\leq)
\end{array}
\]
where either $P\in \ValTerm$, $\omega \equiv \omega_\Val$, $\sigma,\sigma' \in \ValType$ and $\leq \,=\, \leq_\Val$ or
$P\in \ComTerm$, $\omega \equiv \omega_\Comp$, $\sigma,\sigma' \in \ComType$ and $\leq \,=\, \leq_\Comp$.
\end{definition}

In the following we write $\Gamma \der V: \delta$ and $\Gamma \der M: \tau$ to mean that 
these judgments are derivable. The next Lemma is an extension of the analogous property of
BCD type system, also called Inversion Lemma in \cite{BarendregtDS2013} 14.1.

\begin{lemma}[Generation lemma] \label{lem:genLemma}
Assume that $\delta \neq \omega_\Val$ and $\tau \neq \omega_\Comp$, then:
\begin{enumerate}
\item $\Gamma \der x: \delta \Then \Gamma(x) \leq_\Val \delta$
\item $\Gamma \der \lambda x.M: \delta \Then
		\exists I, \delta_i, \tau_i.~  \forall i \in I.\; \Gamma ,x:\delta_i \der M:\tau_i \And  \bigwedge_{i\in I}\delta_i \to \tau_i \leq_\Val \delta$
\item $\Gamma \der \Unit V : \tau \Then \exists \delta.\; \Gamma \der V:\delta \And T\delta \leq_\Comp \tau$
\item $\Gamma \der M\Bind V : \tau \Then$ \\
       		$\exists I, \delta_i, \tau_i.~  \forall i \in I.\; \Gamma \der M:T\delta_i \And \Gamma\der V:\delta_i\to\tau_i \And
			\bigwedge_{i\in I}\tau_i \leq_\Comp \tau$
\end{enumerate}
\end{lemma}

\begin{lemma}[Substitution lemma] \label{lem:SubstLemma} 

If $\Gamma, x:\delta \der M:\tau$ and $\Gamma \der V:\delta$ then $\Gamma \der M\Subst{V}{x}: \tau$.
\end{lemma}

\begin{theorem}[Subject reduction]\label{thm:subjectReduction}

If $\Gamma \der M:\tau$ and $M \Red N$ then $\Gamma \der N:\tau$.
\end{theorem}

\Proof
We only consider the case $M\equiv \Unit V \Bind (\lambda x. M')$ and $N\equiv M'\Subst{V}{x}$.
	Since $\Gamma \der M:\tau $, by \ref{lem:genLemma} we have:
	$
	\exists I, \delta_i, \tau_i. \ \forall i\in I \ \Gamma \der \Unit V: T\delta_i \mbox{ and } 
	\Gamma \der \lambda x.M':\delta_i \to \tau_i \mbox{ and } \bigwedge_{i\in I} \tau_i \leq \tau 
	$.
	By the same lemma we also have
	$
	\exists I, \delta_i, \tau_i. \ \forall i\in I\   \exists \delta'_i\  \Gamma \der V: \delta'_i \mbox{ and } 
	T\delta'_i\leq T\delta_i \mbox{ and }  \Gamma , x:\delta_i  \der M':\tau_i \mbox{ and } \bigwedge_{i\in I} \tau_i \leq \tau $
	By $\IntrInter$ and $(\leq)$ rules, we get $\Gamma , x:\delta \der M':\tau$, and hence by 
	Lemma \ref{lem:SubstLemma} we have that $\Gamma \der M'\Subst{V}{x} : \tau$.
\QED

\section{Type interpretation and soundness}\label{sec:TypeInterp}

To interpret  value and computation types we extend the usual interpretation of intersection
types over $\lambda$-models to $T$-models. Let $(D, T, \Phi, \Psi)$ 
be such a model; then for $d, d' \in D$ we abbreviate: $d \cdot d' = \Phi(d)(d')$. Also, if $X\subseteq D$ and $Y\subseteq TD$ then $X\Rightarrow Y=\Set{d\in D\mid \forall d'\in X.\; d \cdot d'\in Y}$.

\begin{definition}\label{def:interpretationtypes}
	Let  $\xi \in \TypeEnv_D = \TypeVar \to  2^D$ a type variable interpretation; 
	the maps $\Sem{\cdot}^D : \ValType \times \TypeEnv_D \to 2^D $ and 
	$\Sem{\cdot}^{TD} : \ComType \times \TypeEnv_D \to 2^{TD}$ 
	are {\em type interpretations} if:
	\[\begin{array}{rcl@{\hspace{1cm}}rcl}
	\Sem{\alpha}^D_\xi & = & \xi(\alpha) & 
	\Sem{\delta \to \tau}^D_\xi & = &  \Sem{\delta}^D_\xi \Rightarrow \Sem{\tau}^{TD}_\xi  \\ [2mm]
	\Sem{\omega_\Val}^D_\xi & = & D  & 
	\Sem{\delta \wedge \delta'}^D_\xi & = &\Sem{\delta}^D_\xi \cap \Sem{\delta'}^D_\xi \\ [2mm]
	\Sem{\omega_\Comp}^{TD}_\xi & = & TD &
	\Sem{\tau \wedge \tau'}^{TD}_\xi & = & \Sem{\tau}^{TD}_\xi \cap \Sem{\tau'}^{TD}_\xi 
	\end{array}
	\]
	Moreover the following implication holds:
	\[d \in \Sem{\delta}^D_\xi \Then \Unit d \in \Sem{T\delta}^{TD}_\xi\]
	We say that $\Sem{\cdot}^{TD}$ is {\em strict} if 
	\[\Sem{T\delta}^{TD}_\xi = \Set{\Unit d \mid d \in \Sem{\delta}_\xi}\]
\end{definition}

In the following if $D$ is a $T$-model then we assume that $\Sem{\cdot}^D$ and 
$\Sem{\cdot}^{TD}$ are type interpretations satisfying all conditions in Definition \ref{def:interpretationtypes} but are not necessarily strict.
Also we assume $\xi \in \TypeEnv_D$ is admissible and $\rho \in \Env_D$. 

\begin{lemma}\label{lem:inclusion}
	Let $D$,
	$\Sem{\cdot}^{TD}$, and $\xi \in \TypeEnv_D$ as above.\\
	The couple $(D,\xi)$ \emph{preserves} $\leq_\Val$ and $\leq_\Comp$, that is: for all $\delta, \delta' \in \ValType$ and for all $\tau, \tau'\in \ComType$, one has:
	\[ \delta \leq_\Val \delta' ~ \Then ~ \Sem{\delta}^D_\xi \subseteq \Sem{\delta'}^D_\xi 
	\quad \mbox{and} \quad
	\tau \leq_\Comp \tau' ~ \Then ~ \Sem{\tau}^{TD}_\xi \subseteq \Sem{\tau'}^{TD}_\xi 
	\]
\end{lemma}

Clearly any strict type interpretation is a type interpretation. Let us define (strict) truth and (strict) validity of typing judgments. 

\begin{definition}[Truth and validity]\label{def:judgmentModel}
	We say that $\Gamma \der V:\delta$ and $\Gamma \der M:\tau$ are {\em true} in a $T$-model $D$ w.r.t. the type
	interpretations $\Sem{\cdot}^D$ and $\Sem{\cdot}^{TD}$ if
	\begin{enumerate}
		\item $\rho,\xi \models^D \Gamma$ if $\rho(x) \in \Sem{\Gamma(x)}^D_\xi$ for all $x \in \Dom(\Gamma)$
		\item $\Gamma \models^D V:\delta$ if $\rho,\xi \models^D \Gamma$ implies
			$\Sem{V}^D_\rho \in \Sem{\delta}^D_\xi$ 
		\item $\Gamma \models^D M:\tau$ if $\rho,\xi \models^D \Gamma$ implies $\Sem{M}^{TD}_\rho \in \Sem{\tau}^{TD}_\xi$
	\end{enumerate}
	We say that $\Gamma \der V:\delta$ and $\Gamma \der M:\tau$ are {\em valid}, written $\Gamma \models V:\delta$ and $\Gamma \models M:\tau$ 
	respectively, if $\Gamma \models^D V:\delta$ and $\Gamma \models^D M:\tau$ for any $T$-model $D$ and type interpretations 
	$\Sem{\cdot}^D$ and $\Sem{\cdot}^{TD}$.
	
	Finally we say that $\Gamma \der V:\delta$ and $\Gamma \der M:\tau$ are {\em strictly true} in $D$, written $\Gamma \models^D_s V:\delta$
	and $\Gamma \models^D_s M:\tau$ respectively, if the interpretations $\Sem{\cdot}^D$ and $\Sem{\cdot}^{TD}$ are strict;
	also they are {\em strictly valid}, written $\Gamma \models_s V:\delta$ and $\Gamma \models_s M:\tau$, if strictly true for any $D$.
\end{definition}

Unfortunately the type system in Definition \ref{def:typeAssignment} is not sound w.r.t. arbitrary type interpretations. The difficulty comes from rule
$\ElimArr$, because in general $\Sem{T\delta}^{TD}_\xi \supseteq \Set{\Unit d \mid d \in \Sem{\delta}_\xi}$ and inclusion might be proper. 
In case of strict
type interpretations, however, this is an equality and the soundness of type assignment is easily established.

\begin{theorem}[Soundness w.r.t. strict interpretations]\label{thr:strict-soundness}
	\[ \Gamma \der V :\delta ~ \Then ~ \Gamma \models_s V:\delta \quad \mbox{and} \quad
	\Gamma \der M :\tau ~ \Then ~ \Gamma \models_s M:\tau
	\]
\end{theorem}
\Proof
	By simultaneous induction on the derivations of $\Gamma\der V:\delta$ and $\Gamma\der M:\tau$.
	Rules $\IntrInter$ and $(\omega)$ are sound by the definition of type  interpretation.
	Rule $(\leq)$ is sound by Lemma \ref{lem:inclusion}, as we proved that $(D,\xi)$ 
	preserves $\leq_\Val$ and $\leq_\Comp$; rule $\IntrUnit$ is immediate by induction.
	
	\begin{description}
	\item Rule $\IntrArr$: if $\rho,\xi \models^D_s \Gamma$ and $d\in \Sem{\delta}^D_\xi$ then
		$\rho[x\mapsto d], \xi \models^D_s \Gamma , x:\delta$; by induction $\Gamma , x:\delta \models^D_s M:\tau$ which implies
		\[\Sem{\lambda x.M}^D_\rho \cdot d = \Sem{M}^{TD}_{\rho[x\mapsto d]}\in \Sem{\tau}^{TD}_\xi\]
		Hence $\Sem{\lambda x.M}^D_\rho \in \Sem{\delta\to \tau}^D_\xi$ by the arbitrary choice of $d$.
		
	\item Rule $\ElimArr$: let $\rho,\xi \models^D_s \Gamma$ and assume by induction that $\Sem{M}^{TD}_\rho \in \Sem{T\delta}^{TD}_\xi$; 
		 because of strictness we have that  $\Sem{M}^{TD}_\rho = \Unit d$ for some
		 $d \in \Sem{\delta}^D_\xi$, therefore
		 \[ \Sem{M \Bind V}^{TD}_\rho = \Sem{M}^{TD}_\rho \Bind \Sem{V}^{D}_\rho = \Unit d \Bind \Sem{V}^{D}_\rho = \Sem{V}^{D}_\rho \cdot d
		 \in \Sem{\tau}^{TD}_\xi\]
		 since $\Sem{V}^{D}_\rho \in \Sem{\delta\to \tau}^{D}_\xi$ by induction.
	\end{description}
	\QED
	
\newcommand{\StateM}{\textsf{S}}
\newcommand{\Loc}{{\cal L}}
\newcommand{\MLet}[2]{\textsf{let} \, #1 \, \textsf{in} \, #2}

The hypothesis of strictness of type interpretation is quite restrictive (although it suffices for proving computational adequacy: see Section \ref{sec:adequacy}).
Consider for example the state monad $\StateM D = S \to D \times S$, where $S$ is some domain of states (e.g. $S \subseteq \Loc \to D$ where $\Loc$ is some
set of locations, and it is ordered by $s \sqsubseteq_S s'$ if $s(\ell) \sqsubseteq_D s'(\ell)$ for all $\ell \in \Loc$), and according to \cite{Moggi'91} define:
\[ \UnitSub{S} \,d = \metalambda s \in S. (d, s) \qquad  a \Bind_S f = \metalambda s \in S. \; \MLet{(d, s') = a\,s}{f\,d\,s'}\]
Given some $d \in D$ and $\ell \in \Loc$, let $a = \UnitSub{S} \,d \in \StateM D$ and 
\[ f = \metalambda d \, s.\, (d, s[\ell:=d]) \in D \to \StateM D,\] 
where $s[\ell:=d]$ is the update of $s$ in $\ell$ to $d$. Now $a\, s = (d, s)$ and therefore
\[ (a \Bind_S f) s = f \,d \,s = (d, s[\ell:=d]) \]
namely $a \Bind_S f = \metalambda s. (d, s[\ell:=d])$ which is a defined computation that is not $\bot_{\StateM D}$ but different than $\UnitSub{S} \,d$, therefore it doesn't belong to the strict type interpretation of any non trivial type $\StateM\, \delta$.

\begin{definition}[Monadic type interpretation]\label{def:monadicTypeInterp}
Let $D$ be a $T$-model; then type interpretations $\Sem{\cdot}^D$ and $\Sem{\cdot}^{TD}$ are monadic if for any $d \in D$ and $a \in TD$
\[ \exists \delta' .\,d\in \Sem{\delta' \to T\delta}^D_\xi \And a\in \Sem{T\delta'}^{TD}_\xi \Then a\Bind d \in \Sem{T\delta}^{TD}_\xi.\]

\end{definition}

By the very definition monadic type interpretations are not inductive; 
in particular $\Sem{T\delta}^{TD}_\xi$ depends on itself and also on $\Sem{T\delta'}^{TD}_\xi$ 
for any $\delta'$. To turn this definition into an inductive one, we make essential use of the
correspondence among intersection types and the category of $\omega$-algebraic lattices.

Henceforth $\Dcat$ is the category of $\omega$-algebraic lattices, which is a particular category of domains (see e.g. \cite{AbramskyJung94} and \cite{Amadio-Curien'98}, ch. 5). Objects of $\Dcat$
are lattices whose elements are directed sup of the {\em compact points} 
$c \in \Compact(D) \subseteq D$ they dominate, where $c$ is {\em compact} if whenever $X$ is directed and $c \Order \bigsqcup X$ there is some  $x\in X$ s.t. $c \sqsubseteq  x$; moreover $\Compact(D)$ is countable.

Suppose that the monad $T$ is an $\omega$-continuous functor over $\Dcat$, so that the functor $\textbf{F}(X) = X \to TX$ is such.
Let $D_\infty = \invLim D_n$ where $D_0$ is some fixed domain, and 
$D_{n+1} = \textbf{F}(D_n) = [D_n \to TD_n]$ is such that for all $n$, 
$D_n \triangleleft D_{n+1}$ is an embedding. As a consequence we have
$D_\infty \simeq \textbf{F}(D_\infty) = [D_\infty \to TD_\infty]$ and we call $\Phi$ the isomorphism. Also, by continuity of $T$,
we have that $TD_\infty = \invLim TD_n$. If $x \in D_\infty$, we write $x_n$ for its projection to $D_n$; 
if $x \in D_n$ then we identify $x$ with its injection into $D_\infty$. 
We fix the notation for the standard notions from  inverse limit construction.
For all $n\in \mathbb{N}$ the following are injection-projection pairs:
\begin{itemize}
		\item $\varepsilon_n: D_n \to D_{n+1}$, $\pi_n:D_{n+1} \to D_n$
		\item Let $m>n$. $\varepsilon_{n,m}: D_n \to D_{m}$, $\pi_{m,n}:D_{m} \to D_n$
		\item $\varepsilon_{n,\infty}: D_n \to D_{\infty}$, $\pi_{\infty,n}:D_{\infty} \to D_n$
\end{itemize}
For definitions see any standard text on domain theory, e.g. \cite{Amadio-Curien'98}. The following lemma lists some well known facts.
We set the following abbreviation $x_n= \pi_{\infty,n}(x)$.
\begin{lemma}\label{lem:D_inftyPropeties}
Let $x, y \in D_\infty$:

\begin{enumerate}
\item $x = \bigsqcup_n x_n$ 
\item if $x \in D_n$ then $x = x_n$
\item $x \cdot y = \bigsqcup_n x_{n+1}(y_n)$
\item if $y \in D_n$ then $x \cdot y = x_{n+1}(y)$
\end{enumerate}
where $x_n= \pi_{\infty,n}(x)$.
\end{lemma}
To these we add:

\begin{lemma}\label{lem:D_inftyPropeties-T}
Let $x \in D_\infty$ and $y \in TD_\infty$:

\begin{enumerate}
\item $(\Unit x)_n = \Unit_{n+1} x$ 
\item $(y \Bind x)_n = y_n \Bind x_n$
\item $a\Bind d= \bigsqcup_n (a_n \Bind d_{n+1})$.
\end{enumerate}
\end{lemma}

If $X \subseteq D_\infty$ we write $X_n = \Set{x_n \mid x \in X}$, similarly for $Y_n$ when $Y \subseteq TD_\infty$. By means of this we define a notion of approximated type interpretation such that each type
turns out to denote certain well behaved subset either of $D_n$ or $TD_n$ according to its kind, called
admissible subset (\cite{AbramskyJung94} sec. 2.1.6). 

\begin{definition}\label{def:admissiblepred}
	A predicate on (i.e. a subset of) a domain $D$ is \emph{admissible} 
	if it contains $\bot$ and is closed under directed suprema.
	We write as $\Adm(D)$ the set of all admissible predicates on a domain $D$.
	Fix a domain $D$ and a subset $X\subseteq D$; we define the following operator: 
	\[cl_D(X)= \bigcap \Set{P \in \Adm(D) \mid P\supseteq \Compact(X)}\] where 
	$\Compact(X)$ are the elements of $\Compact(D)$ bounded above by the elements of $X$.
\end{definition}

Our goal is to show that the type interpretations are admissible subsets of either $D_n$ or $TD_n$
provided that the $\xi(\alpha)$ are such. To enforce admissibility of the $\xi(\alpha)$ we have introduced the
$cl_D$ operator.

\begin{lemma}\label{lem:admissible}
	The operator $cl_D: \Power D \to \Power D$ is a closure, and $cl_D(X)$ is an object of $\Dcat$ for all 	$X\subseteq D$, hence it is admissible and algebraic. 
\end{lemma}

Next we define the notion of approximated type interpretation, that in the limit is a monadic
type interpretation (Theorem \ref{thr:adequateTypeInterp}).

\begin{definition}\label{def:approximateTypeInterp}
Let $\xi \in \Env_{D_\infty}$; then we define a family of {\em approximated type interpretations} $\Sem{\delta}^{D_n}_\xi \subseteq D_n$ and 
$\Sem{\tau}^{TD_n}_\xi \subseteq TD_n$ inductively over $n \in \Nat$, and then over types:

\[ \Sem{\alpha}^{D_n}_\xi = cl_{D_n}(\xi(\alpha)_n) \qquad \Sem{\omega_\Val}^{D_n}_\xi = D_n \qquad \Sem{\omega_\Comp}^{TD_n}_\xi = TD_n\]
\[ \Sem{\delta \wedge \delta'}^{D_n}_\xi = \Sem{\delta}^{D_n}_\xi \cap \Sem{\delta'}^{D_n}_\xi \qquad
   \Sem{\tau \wedge \tau'}^{TD_n}_\xi = \Sem{\tau}^{TD_n}_\xi \cap \Sem{\tau'}^{TD_n}_\xi \]
\[	\Sem{\delta \to \tau}^{D_0}_\xi = D_0 \qquad \qquad
	\Sem{\delta \to \tau}^{D_{n+1}}_\xi = \Set{d\in D_{n+1} \mid \forall d'\in \Sem{\delta}^{D_n}_\xi .\; d(d')\in \Sem{\tau}^{TD_n}_\xi} \]
\[\begin{array}{lll}
\Sem{T\delta}^{TD_0}_\xi & = & TD_0 \\ [1mm]
\Sem{T\delta}^{TD_{n+1}}_\xi & = & \Set{\Unit d  \in TD_n \mid \, d \in \Sem{\delta}^{D_n}_\xi } \; \cup \\ [1mm]
	&& \Set{a\Bind d \in TD_{n+1} \mid \exists \delta' .\,d\in \Sem{\delta' \to T\delta}^{D_{n+1}}_\xi \And a\in \Sem{T\delta'}^{TD_n}_\xi }
\end{array}
\]
\end{definition}

The use of $cl_{D_n}$ in the definition of $\Sem{\alpha}^{D_n}_\xi$ can be avoided if $\xi(\alpha)_n$ is
admissible, for which it suffices that $\xi(\alpha)$ is admissible. We say that $\xi$ is admissible
if all $\xi(\alpha)$ are such. Clearly $cl_D(\xi(\alpha)) = \xi(\alpha)$ if $\xi$ is admissible.

\begin{lemma}\label{lem:adminter}
Let $\xi\in\Env_{D_\infty}$. For all $n\in \mathbb{N}$, every element of the families $\Sem{\delta}^{D_n}_\xi$ and $\Sem{\tau}^{TD_n}_\xi$ is an admissible predicate on $D_n$ and $TD_n$, respectively. Moreover, they are algebraic domains.
In addition, every such admissible predicate is an $\omega$-algebraic sublattice of $D_n$ and $TD_n$, respectively.
\end{lemma}

\begin{lemma}\label{lem:i-ppair}
	Fix a $\xi\in\Env_{D_\infty}$. Define $\varepsilon^\delta_n = \varepsilon_n \upharpoonright \Sem{\delta}^{D_n}_\xi$ and $\pi^\delta_n = \pi_n \upharpoonright \Sem{\delta}^{D_n}_\xi$. \\
	Then, for all $\delta \in \ValType$ $(\varepsilon^\delta_n, \pi^\delta_n)$ is injection-projection pair.
\end{lemma}

Define $\varepsilon^{T\delta}_n=T\varepsilon^{\delta}_n$ and $\pi^{T\delta}_n=T\pi^{\delta}_n$.
By functoriality of $T$ and Lemma \ref{lem:i-ppair}  they are an injection-projection pair, too.
This implies that the following subsets of $D_\infty$ and $TD_\infty$ do exist.

\begin{definition}\label{def:limitsem}
\hfill

	\begin{enumerate}
		\item $\Sem{\delta}^{D_\infty}_\xi=\lim_{\leftarrow}\Sem{\delta}^{D_n}_\xi$
		\item $\Sem{T\delta}^{TD_\infty}_\xi=\lim_{\leftarrow}\Sem{T\delta}^{TD_n}_\xi$
	\end{enumerate} 
\end{definition}



\begin{theorem}\label{thr:adequateTypeInterp}
Let $\xi\in\Env_{D_\infty}$ be admissible. Then
the mappings  $\Sem{\cdot}^{D_\infty}_\xi$ and $\Sem{\cdot}^{TD_\infty}_\xi$ are monadic type interpretations; in particular:
\begin{enumerate}
	\item $d \in \Sem{\delta}^{D_\infty}_\xi \Then \Unit d \in \Sem{T\delta}^{TD_\infty}_\xi$
	\item $\exists \delta' .\,d\in \Sem{\delta' \to T\delta}^{D_\infty}_\xi \And a\in \Sem{T\delta'}^{TD_\infty}_\xi \Then a\Bind d \in \Sem{T\delta}^{TD_\infty}_\xi$
\end{enumerate}
\end{theorem}
\Proof
	By Definition \ref{def:approximateTypeInterp} $\Sem{\delta \to \tau}^{D_{n+1}}_\xi =  \Sem{\delta}^{D_n}_\xi \Rightarrow \Sem{\tau}^{TD_n}_\xi$. Let $d\in \Sem{\delta \to \tau}^{D_{n+1}}_\xi$ if and only if $d=\bigsqcup_n d_{n+1}$ where $d_{n+1}\in \Sem{\delta \to \tau}^{D_{n+1}}_\xi$ for all $n$.
	Similarly,$d'\in\Sem{\delta}^{D_\infty}_\xi$ if and only if $d'=\bigsqcup_n d'_{n+1}$ where $d'_{n+1}\in \Sem{\delta \to \tau}^{D_{n+1}}_\xi$ for all $n$.
	Now, by Lemma \ref{lem:D_inftyPropeties}, $d\cdot d'=\bigsqcup_n d_{n+1}(d'_n)$ and as seen so far $d_{n+1}(d'_n)\in \Sem{\tau}^{TD_n}_\xi$. Hence, $\bigsqcup_n d_{n+1}(d'_n)\in \lim_{\leftarrow}\Sem{\tau}^{TD_n}_\xi=\Sem{\tau}^{TD_\infty}_\xi$.
	\\ Now we prove the \emph{monadicity} of the type interpretation. Let $d\in \Sem{\delta}^{D_\infty}_\xi$. In particular, $\Sem{\delta}^{D_\infty}_\xi=\lim_{\leftarrow}\Sem{\delta}^{D_n}_\xi$, thus $d= \bigsqcup_n d_n$ where $d_n\in \Sem{\delta}^{D_n}_\xi$. 
	Since $\Unit d_n\in \Sem{T\delta}^{TD_n}_\xi$ by definition, and since $\Unit$ is continuous, $\Unit d=\Unit (\bigsqcup_n d_n)= \bigsqcup_n \Unit d_n$. As built so far, $\bigsqcup_n \Unit d_n \in \lim_{\leftarrow}\Sem{T\delta}^{TD_n}_\xi=\Sem{T\delta}^{TD_\infty}_\xi$. So clause (i) is satisfied.\\
	Let $a\in \Sem{T\delta'}^{TD_\infty}_\xi$ and $d\in \Sem{\delta'\to T\delta}^{D_{n+1}}_\xi$. According to (iii) in Lemma \ref{lem:D_inftyPropeties-T}, we have $a\Bind d= \bigsqcup_n (a_n \Bind d_{n+1})$. By construction $a_n\in \Sem{T\delta'}^{TD_n}_\xi$ and $d_{n+1}\in \Sem{\delta'\to T\delta}^{D_{n+1}}_\xi$. By definition $a_n\Bind d_{n+1}\in \Sem{T\delta}^{TD_{n+1}}_\xi$, thus $a\Bind d\in \lim_{\leftarrow}\Sem{T\delta}^{TD_n}_\xi=\Sem{T\delta}^{TD_\infty}_\xi$. So clause (ii) is satisfied.
\QED

\section{The filter model construction }\label{sec:filter}


Let $(L, \leq)$ be an inf-semilattice; a non empty $F \subseteq L$ is a {\em filter} of $L$ if it is upward closed and closed 
under finite infs; $\Filt(L)$ is the set of filters of $L$. The next proposition and theorem are known from the literature, e.g. \cite{DaveyPriestley90,Amadio-Curien'98}:

\begin{proposition}
If $(L, \leq)$ is an inf-semilattice then $\Filt(L) = (\Filt(L), \subseteq)$ is an algebraic lattice, whose compact elements are the filters $\Up a = \Set{a' \in L \mid a \leq a'}$. Hence $\Filt(L)$ is $\omega$-algebraic if $L$ is denumerable.
\end{proposition}

Any $D \in |\Dcat|$ arises by ideal completion of $\Compact(D)$ taken with the restriction of the order $\Order$ over $D$; dually it is isomorphic to the {\em filter completion} of $\CompactOp(D)$, that is $\Compact(D)$ ordered by $\OrderOp$, the inverse of $\Order$.

\begin{theorem}[Representation theorem]\label{thr:representation}
Let $D \in |\Dcat|$; then $\CompactOp(D)$ is an inf-semilattice and $D \simeq \Filt(\CompactOp(D))$ is an isomorphism in $\Dcat$. 
\end{theorem}

Let $\Th = (\Types, \leq)$ be a type theory. Elements of $\Filt(\TypeTh)$ are the non empty subsets $F \subseteq \Types$ which are upward closed w.r.t. the preorder $\leq$ and such that if $\sigma,\sigma' \in F$ then $\sigma\Inter\sigma'\in F$. This definition, coming from \cite{BCD'83}, is not the same
as that of $\Filt(L)$ because $\leq$ is just a preorder, and infs do exist only in the quotient $\TypeTh_{/\leq}$, 
which is indeed an inf-semilattice. The resulting partial order $\Filt(\Th) = (\Filt(\TypeTh), \subseteq)$, however, is isomorphic to $(\Filt(\TypeTh_{/\leq}), \subseteq)$.


\begin{lemma}\label{lem:typeFilters}
Let $\TypeTh_{/\leq}$ be the quotient of the pre-order $\Th = (\TypeTh, \leq)$, whose elements are the equivalence classes $[\sigma] = \Set{\sigma' \in \TypeTh \mid \sigma \leq \sigma' \leq \sigma}$, ordered by the relation $[\sigma] \lesssim [\sigma'] \Iff \sigma \leq \sigma'$. Then 
$(\TypeTh_{/\leq}, \lesssim)$ is an inf-semilattice; moreover
$\Filt(\Th) \simeq \Filt(\TypeTh_{/\leq})$ is an isomorphism in $\Dcat$.
\end{lemma}

\Proof Proving that $\lesssim$ is well defined and $(\TypeTh_{/\leq}, \lesssim)$ is an inf-semilattice is routine. 
The isomorphism $\Filt(\Th) \simeq \Filt(\TypeTh_{/\leq})$ is given by the map
$F \mapsto \Set{ [\sigma] \mid  \sigma \in F}$. 
\QED

%
By Theorem \ref{thr:representation} and Lemma \ref{lem:typeFilters}, any $D \in |\Dcat|$ is isomorphic to the filter domain $\Filt_D = \Filt(\TypeTh_D)$
of some intersection type theory $\Th_D = (\TypeTh_D, \leq_D)$, called the Lindenbaum algebra of $\Compact(D)$ in \cite{Abramsky'91}.


\begin{definition}\label{def:eats}\hfill

\noindent
For $D, E \in |\Dcat|$, the {\em functional type theory} 
$\Th_{D\to E} = (\TypeTh_{D\to E}, \leq_{D\to E})$ is the least type theory
such that $\TypeTh_{D\to E}$ includes all expressions of the form $\delta \to \varepsilon$ for
$\delta \in \TypeTh_D$ and $\varepsilon \in \TypeTh_E$ and $ \leq_{D\to E}$ is such that
\[ \begin{array}{l@{\hspace{1cm}}c}
\omega_{D \to E} \leq_{D\to E} \omega_D \to \omega_E &
(\delta \to \varepsilon) \Inter (\delta \to \varepsilon') \leq_{D\to E} \delta \to (\varepsilon \Inter \varepsilon')  
\\ [2mm]
\prooftree
	\delta' \leq_{D\to E} \delta \quad \varepsilon \leq_E \varepsilon'
\justifies
	\delta \to \varepsilon \leq_{D\to E} \delta' \to \varepsilon'
\endprooftree
\end{array}\]
Also $\Th_{D\to E}$ is {\em continuous} if
\begin{equation}\label{eq:continuousEats}
\bigwedge_{i\in I}  (\delta_i\to\varepsilon_i) \leq_{D\to E} \delta \to \varepsilon \Then
	\bigwedge \Set{\varepsilon_i \mid i\in I \And \delta \leq_D \delta_i} \leq_E \varepsilon 
\end{equation}
\end{definition}

\begin{remark}\label{rem:eats}
The theory $\Th_{D\to E}$ is an {\em extended abstract type systems}, shortly {\em eats} 
(see e.g. \cite{Amadio-Curien'98} ch. 3),
but for the sorts of type expressions. It is continuous if it is a continuous eats.
\end{remark}

%

\begin{proposition}\label{prop:eats}
Let $\Th_{D\to E}$ be a continuous functional type theory. Then
the domain $\Filt_{D\to E} = \Filt(\Th_{D\to E})$ is isomorphic to $[\Filt_D \to \Filt_E]$, namely the 
domain of Scott continuous functions from $D$ to $E$.
\end{proposition}

\Proof
First if $u \in \Filt_{D\to E}$ and $d \in \Filt_D$ then
\[u \cdot d = \Set{\varepsilon \in \Types_E \mid \exists\, \delta \to \varepsilon \in u. \;\delta \in d} \in \Filt_E\] 
Then the isomorphism $\Filt_{D\to E} \simeq [\Filt_D \to \Filt_E]$ is given by 
\begin{equation}\label{eq:PhiPsi}
\Phi^\Filt(u) = \metalambda d \in \Filt_D.\; u \cdot d \qquad 
\Psi^\Filt(f) = \Set{\bigwedge_{i\in I} (\delta_i\to\varepsilon_i) \mid \forall i \in I.\; \varepsilon_i \in f(\!\Up\delta_i)}
\end{equation}
To prove that  $\Phi^\Filt\comp \Psi^\Filt = \Id_{[\Filt_D \to \Filt_E]}$ and $\Psi^\Filt \comp \Phi^\Filt = \Id_{\Filt_{D\to E}}$ it is enough to show this for compact elements, which is obtained by simple calculations.
\QED

\begin{definition}\label{def:Th^T}
Let $\Th = (\TypeTh, \leq)$ be a type theory and $T$ a unary symbol; then 
$\Th^T = (\TypeTh^T, \leq^T)$ is the least type theory such that $\TypeTh^T$ is defined by the grammar
\[\varphi ::= T\sigma \mid \varphi \Inter \varphi' \mid \omega_T\]
and, for $\sigma, \sigma' \in \TypeTh$ 
\[\varphi \leq^T \omega_T \qquad
\sigma \leq \sigma' \Then T\sigma \leq^T T\sigma' \qquad T\sigma \Inter T\sigma' \leq^T T(\sigma \Inter \sigma')\]
\end{definition}

Clearly $\Th_\Comp = (\Th_\Val)^T$ that we abbreviate by $\Th_\Val^T$; also, by Theorem \ref{thr:representation}, if $\Th_D$ is the theory of some $D\in |\Dcat|$ then 
$\Filt_{TD} = \Filt(\Th_D^T)$ is again a domain with theory $\Th_D^T$, and 
$\TT D = \Filt_{TD}$ is a well defined, total map $\TT: |\Dcat| \to |\Dcat|$. 

In the following we abbreviate $\Filt_D = \Filt(\Th_D)$ and $\Filt_{TD} = \Filt(\Th_D^T)$. We also suppose that the set of atoms $\TypeVar_0$ is non empty and fixed, and 
in one-to-one correspondence to the compacts $\Compact(D_0)$ of some fixed
domain $D_0$; also assume that $\alpha_d \leq_\Val \alpha_e$ in the theory $\Th_\Val$ if and only if $e \sqsubseteq d$ in $D_0$.

Finally, we enforce extensionality of the resulting $T$-model (see below) by adding to $\Th_\Val$
axioms that equate each atomic type $\alpha$ to an arrow type, for which we have similar choices
that is either all (in-)equations $\alpha =_\Val \omega_\Val \to T\alpha$
by analogy to Scott's models, or $\alpha =_\Val \alpha_\Val \to T\alpha$ by analogy to Park's. No matter
which is the actual choice, we ambiguously call $\Th_\Val^\eta$ the resulting theory and
$D_* = \Filt(\Th_\Val^\eta)$ its filter domain.

 \begin{lemma}\label{lem:filtIso}
 We have that $\Th_\Val^\eta = \Th_{D_* \to \TT D_*}$ and it is
 a continuous functional theory. Therefore $D_* \simeq [D_* \to \TT D_*]$.
 \end{lemma}
 
 \Proof 
By hypothesis $\Th^\eta_\Val = \Th_{D_*}$; by definition  $\Th_\Comp = (\Th^\eta_\Val)^T = \Th_{\TT D_*}$.
By Definition \ref{def:type-theories-Th_V-Th_C}  we have that $\Th_{D_*} = \Th_{D_* \to \TT D_*}$.
To see that $\Th_{D_* \to \TT D_*}$ is a continuous functional type theory it suffices to show
that condition (\ref{eq:continuousEats}) holds in $\Th_\Val$ (similar to the analogous proof for system BCD).
Hence 
\[D_* \simeq \Filt(\Th_{D_* \to \TT D_*}) \simeq [\Filt(\Th_{D_*}) \to \Filt(\Th_{\TT D_*})]
\simeq [D_* \to \TT D_*]
\]
by Proposition \ref{prop:eats}. 
\QED

\begin{lemma}\label{lem:filtMonad}
Let $\TT: |\Dcat| \to |\Dcat|$ be as above. Define $\UnitSub{D}^\Filt: \Filt_D \to \Filt_{TD}$ and 
$\Bind_{D,E}^\Filt: \Filt_{TD} \times \Filt_{D \to \TT E} \, \to \Filt_{\TT E}$ such that:
\[\UnitSub{D}^\Filt \ d = \Up \Set{T\delta \in \Types_{\TT D} \mid \delta \in d} \qquad
t \Bind_{D,E}^\Filt e = \Up\Set{\tau \in \Types_{\TT E} \mid \exists\ \delta\to\tau \in e. \; T\delta \in t}\]
Then $(\TT, \Unit^\Filt, \Bind^\Filt)$ is a monad over $\Dcat$.
\end{lemma}

\begin{corollary}\label{cor:filtModel}
Let $\Phi^\Filt, \Psi^\Filt$ be defined by (\ref{eq:PhiPsi}). 
Then $(D_*, \TT, \Phi^\Filt, \Psi^\Filt)$ is a $T$-model.
\end{corollary}

Next we show that $D_*$ is a limit $T$-model, hence it admits 
a monadic type interpretation by Theorem \ref{thr:adequateTypeInterp}.
Let's stratify types according to the rank map: $r(\alpha) = r(\omega_\Val) = r(\omega_\Comp) = 0$, $r(\sigma\Inter\sigma') = \max(r(\sigma), r(\sigma'))$ (for
$\sigma\Inter\sigma' \in \ValType \cup \ComType$),
$r(\delta\to \tau) = \max(r(\delta)+1, r(\tau))$ and $r(T\delta) = r(\delta) + 1$. If $\Types$ is any language of intersection types, we set: 
$\Types_n = \Set{\sigma \in \Types \mid r(\sigma) \leq n}$.


\begin{lemma}\label{lem:limit}
 Let $\leq_n = ~ \leq \ \upharpoonright \! \Types_n \times \Types_n$ (for both $\leq_\Val$ and $\leq_\Comp$) and $\Th_n$ and 
 $D_n = \Filt(\Th_n)$ be the respective type theories and filter domains. Then the $D_n$ form a denumerable chain of domains
 such that $D_* = \lim_{\leftarrow} D_n$ is a limit $T$-model. 
 \end{lemma}

 \Proof 
Recall that $D_n = \Filt(\Th_n^\eta)$. Define $\varepsilon_n: D_n \to D_{n+1}$, $\pi_n:D_{n+1} \to D_n$ as follows:
\[\begin{array}{rcl@{~~}r}
\varepsilon_n(d) & = &  \Set{\delta \in \ValType_{n+1} \mid \exists \delta' \in d.\; \delta' \leq_{n+1} \delta} & \mbox{for $d \in D_n = \Filt(\Th_n^\eta)$}
\\ [1mm]
\pi_n(e) & = & \Set{\delta \in e \mid r(\delta) \leq n} & \mbox{for $e \in D_{n+1} = \Filt(\Th_{n+1}^\eta)$}
\end{array}
\]
Let $d \in D_n$ and $e = \varepsilon_n(d)$. 
Both $\varepsilon_n(d)$ an $\pi_n(e)$ are filters.  Clearly $d \subseteq \varepsilon_n(d) \cap \ValType_n$, hence
$d \subseteq \pi_n(e)$.

Vice versa if $\delta \in \pi_n(e)$ then $r(\delta) \leq n$ and for some $\delta' \in d$, $\delta' \leq_{n+1} \delta$.
But $\leq_{n+1} \upharpoonright \! \ValType_n \times \ValType_n = ~ \leq_n$, hence
$\delta' \leq_{n} \delta$ which implies $\delta \in d$ as the latter is a filter. We conclude that
$\pi_n(e) = d$, and therefore $\pi_n \circ \varepsilon_n = \Id_{D_n}$ by the arbitrary choice of $d$.

On the other hand if $d = \pi_n(e)$ for some arbitrary $e \in D_{n+1}$ and $\delta \in \varepsilon_n(d)$
it follows that there exists $\delta' \in e$ such that $r(\delta') \leq n$ and $\delta' \leq_{n+1} \delta$, which
implies that $\delta \in e$ being $e$ a filter w.r.t. $\leq_{n+1}$. 

In conclusion $\varepsilon_n \circ \pi_n \Order \Id_{D_{n+1}}$ where $\Order$ is the point-wise ordering induced by subset inclusion. Combining with the previous equation we conclude that $(\varepsilon_n,\pi_n)$
is an injection-projection pair, and hence $D_* \subseteq \lim_{\leftarrow}D_n$. To see the inverse inclusion
just note that any filter $d \in D_*$ is the union of all its restrictions to the rank $n$, namely in $D_n$, and
hence is a filter in $D_\infty$. 
\QED
 
\begin{remark}\label{rem:limit}
It is not the case that $D_n \subseteq D_{n+1}$: indeed let $\delta \to \tau \in d \in D_n$ and take
any $\delta'$ such that $r(\delta') = n$; hence $r(\delta \Inter \delta' \to \tau) = n+1$ and $\delta \Inter \delta' \to \tau \not\in d$. But $\delta \to \tau \leq_{n+1} \delta \Inter \delta' \to \tau$, therefore the latter type belongs to any filter in $D_{n+1}$ including $d$. Then we see that $\varepsilon_n$ cannot be just set theoretic inclusion.
The very same example shows that $\pi_n(e) \subseteq e$ is a proper inclusion in general.
\end{remark}



\section{Soundness and completeness of the type system}\label{sec:soundandcompleteness}

In the following let us fix a $T$-model $D$ and assume that 
$\Sem{\cdot}^{TD}$ is a monadic type interpretation.
Also we assume $\xi \in \TypeEnv_D$ is admissible and $\rho \in \Env_D$. 

\begin{lemma}\label{lem:inclusion}
	Let $D$,
	$\Sem{\cdot}^{TD}$, and $\xi \in \TypeEnv_D$ as above.\\
	The couple $(D,\xi)$ \emph{preserves} $\leq_\Val$ and $\leq_\Comp$, that is: for all $\delta, \delta' \in \ValType$ and for all $\tau, \tau'\in \ComType$, one has:
	\[ \delta \leq_\Val \delta' ~ \Then ~ \Sem{\delta}^D_\xi \subseteq \Sem{\delta'}^D_\xi 
	\quad \mbox{and} \quad
	\tau \leq_\Comp \tau' ~ \Then ~ \Sem{\tau}^{TD}_\xi \subseteq \Sem{\tau'}^{TD}_\xi 
	\]
\end{lemma}

\begin{definition}\label{def:judgmentModel}
	\hfill
	
	\begin{enumerate}
		\item $\rho,\xi \models^D \Gamma$ if $\rho(x) \in \Sem{\Gamma(x)}^D_\xi$ for all $x \in \Dom(\Gamma)$
		\item $\Gamma \models^D V:\delta$ ($\Gamma \models^D M:\tau$) if $\rho,\xi \models^D \Gamma$ implies
		$\Sem{V}^D_\rho \in \Sem{\delta}^D_\xi$ ($\Sem{M}^{TD}_\rho \in \Sem{\tau}^{TD}_\xi$)
		\item $\Gamma \models V:\delta$ ($\Gamma \models M:\tau$) if $\Gamma \models^D V:\delta$ ($\Gamma \models^D M:\tau$) for all $D$. 
	\end{enumerate}
\end{definition}


\begin{theorem}[Soundness]\label{thr:soundness}
	\[ \Gamma \der V :\delta ~ \Then ~ \Gamma \models V:\delta \quad \mbox{and} \quad
	\Gamma \der M :\tau ~ \Then ~ \Gamma \models M:\tau
	\]
\end{theorem}
\Proof
	By simultaneous induction on the derivations of $\Gamma\der V:\delta$ and $\Gamma\der M:\tau$.
	Rules $\IntrInter$ and $(\omega)$ are sound by a quick inspection on the definition 
	of type  interpretation.
	Rule $(\leq)$ is sound by lemma \ref{lem:inclusion}, as we proved that $(D,\xi)$ 
	preserves $\leq_\Val$ and $\leq_\Comp$.
	
	In order to prove soundness of tule $\IntrArr$, assume $\Gamma , x:\delta\der M:\tau$ and $\rho,\xi \models \Gamma$, one has to show that $\Sem{\lambda x.M}^D_\rho \in \Sem{\delta \to \tau}^D_\xi$. Let $d\in \Sem{\delta}^D_\xi$, the thesis is equivalent to say that $\Sem{M}^{TD}_{\xi[x\mapsto d]}\in \Sem{\tau}^{TD}_\xi$, by the arbitrary choice of $d$, because in any model this is the same as $\Sem{\lambda x.M}^D_\rho \cdot d \in \Sem{\tau}^{TD}_\xi$. If we set $F_D=\Set{ e\in D \mid \exists M, \rho .\; e=\Sem{\lambda x.M}^D_\rho}$ then the semantic interpretation of rule $\IntrArr$ says that $\Set{e\in F_D\mid \forall d.; d\in \Sem{\delta}^D_\xi \Rightarrow e\cdot d\in \Sem{\tau}^{TD}_\xi}\subseteq \Sem{\delta\to \tau}^d_\xi$.\\
	For what concerns rule $\IntrUnit$, suppose that the assertion is true for the derivation of $\Gamma \der V:\delta$, namely $\Gamma\models V:\delta$, and $\rho, \xi\models \Gamma$. To prove that $\Gamma \models \Unit V$, one has to show that $\Sem{\Unit V}^{TD}_\rho \in \Sem{T\delta}^{TD}_\xi$, that is equivalent to $\Unit \Sem{V}^{D}_\rho \in \Sem{T\delta}^{TD}_\xi$ and this true by  (i) of Theorem
	\ref{thr:adequateTypeInterp}
Soundness of rule $\ElimArr$ follows a very similar path, invoking propriety (ii) of of the same theorem.
\QED

Recall that $D_* = \Filt(\Th_\Val^\eta)$. 
Let $\xi_0 \in \TypeEnv_{D_*}$ be defined by $\xi_0(\alpha) = \Set{d \in D_* \mid \alpha \in d}
\cup \Set{\!\Up\omega_\Val}$, that is admissible.
Also let $\Sem{\cdot}^{D_*}$ and $\Sem{\cdot}^{TD_*}$ be the monadic interpretations of
Definition \ref{def:limitsem}.

\begin{lemma}\label{lem:xi_0}
\[\Sem{\delta}^{D_*}_{\xi_0} = \Set{d \in D_* \mid \delta \in d} \quad \mbox{and} \quad
\Sem{\tau}^{TD_*}_{\xi_0} = \Set{a \in TD_* \mid \tau \in a}.
\]
\end{lemma}


\Proof $D_*$ is a $T$-model by Corollary \ref{cor:filtModel}, and $\Sem{\cdot}^{D_*}$ and $\Sem{\cdot}^{TD_*}$ are  monadic, hence Lemma \ref{lem:inclusion} applies.
\QED

\begin{lemma}[Type Semantics Theorem]\label{lem:type-semantics}
For any $\rho \in \Env_{D_*}$:

\begin{enumerate}
\item $\Sem{V}^{D_*}_\rho = 
	\Set{\delta \in \ValType \mid \exists \Gamma. \; \rho,\xi_0 \models \Gamma \And \Gamma \der V:\delta}$
\item $\Sem{M}^{TD_*}_\rho = 
	\Set{\tau \in \ComType \mid \exists \Gamma. \; \rho,\xi_0 \models \Gamma \And \Gamma \der M:\tau}$
\end{enumerate}
\end{lemma}

\Proof
By Theorem \ref{thr:soundness} and the fact that $D_*$ is a $T$-model and type interpretations are monadic, both inclusions $\supseteq$ follow by Lemma \ref{lem:xi_0}. To see inclusions $\subseteq$ we reason by induction over $V$ and $M$.
\begin{description}
\item Case $V \equiv x$: if $\delta \in \Sem{x}^{D_*}_\rho = \rho(x)$ take $\Gamma = x:\delta$. Then
	clearly $\Gamma\der x:\delta$. On the other hand
	$\rho,\xi_0 \models x:\delta$ if $\rho(x) \in \Sem{\delta}^{D_*}_{\xi_0}$ 
	that is if $\delta \in \rho(x)$ by Lemma \ref{lem:xi_0}, which holds by hypothesis.
\item Case $V \equiv \lambda x.M$: if $\delta \in \Sem{\lambda x.M}^{D_*}_\rho$, recall 
	from Definition \ref{def:interpretation} and Proposition \ref{prop:eats} that:
	\[\begin{array}{lll}
	\Sem{\lambda x.M}^{D_*}_\rho & = &
	\Psi^\Filt(\metalambda d. \; \Sem{M}^{TD_*}_{\rho[x \,\mapsto d]})  \\
	& = &
	 \Up \Set{\bigwedge_{i\in I} (\delta_i\to\tau_i) \mid \forall i \in I.\; \tau_i \in 
		\Sem{M}^{TD_*}_{\rho[x \,\mapsto \!\Up\,\delta_i]}}
	\end{array}
	\]
	By induction for all $i \in I$ there exists $\Gamma_i$ s.t. 
	$\rho[x \,\mapsto \!\Up\,\delta_i],\xi_0 \models \Gamma_i$ and
	$\Gamma_i \der M : \tau_i$. This implies that $\rho(x) \in \Up\,\delta_i$ hence
	there is no theoretical loss in supposing that $\Gamma_i = \Gamma'_i, x:\delta_i$.
	Let $\Gamma' = \bigwedge_{i\in I}\Gamma'_i$ be the pointwise intersection of the $\Gamma'_i$;
	it follows that $\Gamma', x:\delta_i \der M : \tau_i$ for all $i\in I$. Therefore by $\IntrArr$ we have
	$\Gamma' \der \lambda x.M : \delta_i \to \tau_i$ for all $i\in I$, so that by $\IntrInter$ we conclude
	that $\Gamma' \der \lambda x.M : \bigwedge_{i\in I} (\delta_i\to\tau_i)$, and the thesis follows
	by $\bigwedge_{i\in I} (\delta_i\to\tau_i) \leq_\Val \delta$ and rule $(\leq)$.
\item Case $M \equiv \Unit V$: if $\tau \in \Sem{\Unit V}^{TD_*}_\rho = \Unit \Sem{V}^{D_*}_\rho
	= \Up\Set{T\delta \mid \delta \in \Sem{V}^{D_*}_\rho}$. By induction
	there is $\Gamma$ such that $\rho,\xi_0 \models \Gamma$ and $\Gamma \der V:\delta$, from which
	it follows that $\Gamma \der \Unit V:T\delta$ by $\IntrUnit$.
\item Case $M \equiv M' \Bind V$: if $\tau \in \Sem{M' \Bind V}^{TD_*}_\rho$ where:
	\[\Sem{M' \Bind V}^{TD_*}_\rho = \Sem{M'}^{TD_*}_\rho \Bind \Sem{V}^{D_*}_\rho =
	\Up \Set{\tau \mid \exists \delta \to \tau \in \Sem{V}^{D_*}_\rho.\; T\delta \in \Sem{M'}^{TD_*}_\rho}
	\]
	By induction there exist $\Gamma'$ and $\Gamma''$ such that:
	\[\rho,\xi_0 \models \Gamma' \And \Gamma' \der M':T\delta \quad \mbox{and} \quad
	\rho,\xi_0 \models \Gamma'' \And \Gamma'' \der V:\delta\to\tau
	\]
	Now let $\Gamma = \Gamma' \Inter \Gamma''$; we have that $\rho,\xi_0 \models \Gamma'$
	implies that $\Gamma'(x) \in \rho(x)$ for all $x \in \Dom(\Gamma')$ by Lemma \ref{lem:xi_0},
	and similarly  $\Gamma''(y) \in \rho(y)$ or all $y \in \Dom(\Gamma'')$; hence
	for all $z \in \Dom(\Gamma) = \Dom(\Gamma') \cup \Dom(\Gamma'')$ we have
	$\Gamma(z) = \Gamma'(z) \Inter \Gamma''(z) \in \rho(z)$ since $\rho(z)$ is a filter.
	It follows that $\rho,\xi_0 \models \Gamma$ and, since $\Gamma \leq_\Val \Gamma',\Gamma''$
	that both $\Gamma \der M':T\delta$ and $\Gamma \der V:\delta\to\tau$, from which we obtain
	$\Gamma \der  M' \Bind V : \tau$ by $\ElimArr$.
\end{description}

\QED

\begin{theorem}[Completeness]\label{thr:completeness}
\[ \Gamma \models V:\delta  ~ \Then ~  \Gamma \der V :\delta \quad \mbox{and} \quad
	\Gamma \models M:\tau ~ \Then ~ \Gamma \der M :\tau. 
\]
\end{theorem}

\Proof
We show the second implication as the first one is similar. Assume that $\Gamma \models M:\tau$,
then in particular we have $\Gamma \models^{D_*} M:\tau$. Let $\rho_\Gamma \in 
\Env_{D_*}$ be defined by $\rho_\Gamma(x) = \Up\Gamma(x)$. By construction, we
have $\rho_\Gamma,\xi_0 \models \Gamma$ and hence $\Sem{M}^{TD_*}_{\rho_\Gamma}\in \Sem{\tau}^{TD_*}_{\xi_0}$. Thus, $\tau \in \Sem{M}^{TD_*}_{\rho_\Gamma}$ by Lemma \ref{lem:xi_0}. Therefore, there exists $ \Gamma'$ such that $ \rho_\Gamma, \xi_0 \models \Gamma' \mbox{ and } \Gamma' \der M:\tau$ by Lemma \ref{lem:type-semantics}. 
\\Without loss of theoretical generality, we can assume $X:=\Dom \Gamma=\Dom \Gamma' = \FV (M)$. For all $x\in X$, $\Gamma'(x)\in\rho_\Gamma (x)=\Up\Gamma(x)$ implies that $\Gamma\leq_\Val \Gamma'$. From this last consideration we conclude that $\Gamma \der M:\tau$.

\QED

\begin{corollary}[Subject expansion]\label{cor:subjectExpansion}
If $\Gamma \der M:\tau$ and $N \Red M$ then $\Gamma \der N:\tau$.
\end{corollary}
\Proof
Since $M\Red N$, for all model $D$ and every $\rho$, $\Sem{N}^D_\rho=\Sem{M}^D_\rho$ by Proposition \ref{prop:soundness-of-interpretation}. In particular, by assuming $D=D_*$, $\Sem{N}^{TD_*}_{\rho_\Gamma}=\Sem{M}^{TD_*}_{\rho_\Gamma}$. 
\[
\begin{array}{rcll}
\Gamma \der M:\tau & \Then & \tau \in \Sem{M}^{TD_*}_{\rho_\Gamma} & \mbox{by Theorem \ref{thr:soundness} and Lemma \ref{lem:xi_0}}\\ [1mm]
& \Then & \tau \in \Sem{N}^{TD_*}_{\rho_\Gamma}\\[1mm]
& \Then & \exists \Gamma' .\; \rho_{\Gamma},\xi_0 \models \Gamma' \mbox{ and } \Gamma' \der M:\tau & \mbox{by Lemma \ref{lem:type-semantics}}\\[1mm]
& \Then & \Gamma \der N:\tau & \mbox{as in proof of Theorem \ref{thr:completeness}}.

\end{array}
\] 
\QED

\section{Computational adequacy}\label{sec:adequacy}
In denotational semantics using domain theoretic models, computational adequacy is the property that exactly all divergent programs are interpreted by $\bot$.
This notion is relative to the operational semantics at hand and to the observational notion of convergence. For this to make sense in the present setting,
we introduce a notion of convergence inductively, and relate it to the reduction relation.

Programs are closed expressions; let $Term^0 = \ValTerm^0 \cup \ComTerm^0$ be the set of closed terms.

\begin{definition}\label{def:convergence}
	Let $\Downarrow \; \subseteq \ComTerm^0 \times \ValTerm^0$ be the smallest relation satisfying: 
	\[
	\begin{array}{c@{\hspace{1cm}}c}
		\prooftree 
		\vspace{0.3cm}
		\justifies		\Unit V \Downarrow V
		\endprooftree &
		\prooftree 
		M\Downarrow V \qquad N\Subst{V}{x}\Downarrow W
		\justifies
		M\Bind \lambda x.N\Downarrow W
		\endprooftree 
	\end{array}
	\]
\end{definition}

\begin{lemma}\label{lem:ConvRed}
For all $M \in \ComTerm^0$ and $V \in \ValTerm^0$ we have: 
\[M \Downarrow V \Iff M \RedStar \Unit V.\]
\end{lemma}

\Proof Both implications are proved by straightforward inductions. Just observe that reduction doesn't introduce new free variables, so that
if $M$ is closed then $\Unit V$ and therefore $V$ are such.
\QED

When $M\Downarrow V$ we say that $M$ {\em converges} to $V$. We abbreviate $M\!\Downarrow \;\Iff\; \exists V.\; M\Downarrow V$ and say simply that
$M$ converges. The predicate $M\!\Downarrow$ is non trivial. Indeed consider the closed term:
\[\Omega_\Comp \equiv \Unit (\lambda x. \Unit x \Bind x) \Bind  (\lambda x. \Unit x \Bind x)
\]
that is a translation of the well known term $\Omega \equiv (\lambda x.xx)(\lambda x.xx)$ from ordinary $\lambda$-calculus. Then the only reduction out of
$\Omega_\Comp$ is
\[\Omega_\Comp  \Red  (\Unit x \Bind x)\Subst{\lambda x.\Unit x \Bind x}{x} \equiv \Omega_\Comp \qquad \mbox{by $\beta_c$}
\]
which is not of the shape $\Unit V$ for any $V\in \ValTerm$, hence $\Omega_\Comp \not\Downarrow$.

\medskip
We say that $\tau \in \ComType$ is {\em non trivial} if $\tau \neq_\Comp \omega_\Comp$. The goal is to prove the next theorem, characterizing convergent terms:

\begin{theorem}[Characterization]\label{thr:Conv-nonTrivial}
For all $M \in \ComTerm^0$ we have: 
\[M\!\Downarrow~\Iff~ \exists \,\mbox{non trivial}\ \tau .\, \der M:\tau.\]
\end{theorem}

Toward the proof, and following the pattern of Tait's computability method, we introduce some auxiliary notions.

\begin{definition}\label{def:I-interp}
Let $\I:\TypeVar \to \Power \ValTerm^0$ be a map; then define $\Iinterp{\delta}_\I \subseteq \ValTerm^0$ and $\Iinterp{\tau}_\I \subseteq \ComTerm^0$ by induction
as follows:
\begin{enumerate}
\item $\Iinterp{\alpha}_\I = \I(\alpha)$
\item $\Iinterp{\delta \to \tau}_\I = \Set{V \in \ValTerm^0 \mid \forall M \in \Iinterp{T\delta}_\I.\; M \Bind V \in \Iinterp{\tau}_\I}$
\item $\Iinterp{T\delta}_\I = \Set{M \in \ComTerm^0 \mid \exists V \in \Iinterp{\delta}_\I. \, M \Downarrow V}$
\item $\Iinterp{\omega_\Val}_\I = \ValTerm^0$ and $\Iinterp{\omega_\Comp}_\I = \ComTerm^0$
\item $\Iinterp{\delta \Inter \delta'}_\I = \Iinterp{\delta}_\I \cap \Iinterp{\delta'}_\I$ and 
	$\Iinterp{\tau \Inter \tau'}_\I = \Iinterp{\tau}_\I \cap \Iinterp{\tau'}_\I$.
\end{enumerate}
\end{definition}

\begin{lemma}\label{lem:leq-Iinterp}
Let $\I$ be arbitrary. Then:
\begin{enumerate}
\item $\delta \leq_\Val \delta' \Then \Iinterp{\delta}_\I \subseteq  \Iinterp{\delta'}_\I$
\item $\tau \leq_\Comp \tau' \Then \Iinterp{\tau}_\I \subseteq  \Iinterp{\tau'}_\I$
\end{enumerate}
\end{lemma}

\Proof By checking axioms and rules in Definition \ref{def:type-theories-Th_V-Th_C}. The only non trivial cases concern the arrow and $T$-types.

\begin{description}
\item Let $V \in \Iinterp{(\delta \to \tau_1) \Inter (\delta\to \tau_2) }_\I = \Iinterp{\delta \to \tau_1}_\I \cap   \Iinterp{\delta \to \tau_2}_\I$, then for all
	$M \in \Iinterp{T\delta}_\I$ we have $M \Bind V \in \Iinterp{\tau_i}_\I$ for both $i=1,2$; hence 
	$M \Bind V \in \Iinterp{\tau_1}_\I \cap  \Iinterp{\tau_2}_\I =  \Iinterp{\tau_1 \Inter \tau_2}_\I$.  
	
\item Suppose that $\delta_1 \leq_\Val \delta_2$ and let $M \in \Iinterp{T\delta_1}_\I$; then there exists $V \in \Iinterp{\delta_1}_\I$ such that $M \Downarrow V$.
	By induction $\Iinterp{\delta}_\I \subseteq  \Iinterp{\delta'}_\I$ so that immediately we have $M \in \Iinterp{T\delta_2}_\I$.
	
\item Let $M \in \Iinterp{T\delta_1 \Inter T\delta_2}_\I = \Iinterp{T\delta_1}_\I \cap \Iinterp{T\delta_2}_\I$. Then there exists $V_1 \in \Iinterp{\delta_1}_\I$ and
	$V_2 \in \Iinterp{\delta_2}_\I$ such that $M \Downarrow V_1$ and $M \Downarrow V_2$. By Lemma \ref{lem:ConvRed} we have
	$M \RedStar \Unit V_i$ for both $i=1,2$ and these terms are in normal form; hence $V_1 \equiv V_2$ by Theorem \ref{subsec:confluence}. It follows that
	there exists a unique $V \in \Iinterp{\delta_1}_\I \cap \Iinterp{\delta_2}_\I = \Iinterp{\delta_1 \Inter \delta_2}_\I$ such that $M \RedStar \Unit V$, hence
	$M \in \Iinterp{T(\delta_1\Inter\delta_2)}_\I$.
	
\item Suppose that $\delta_2 \leq_\Val \delta_1$ and $\tau_1 \leq_\Comp \tau_2$. Let $V \in \Iinterp{\delta_1 \to \tau_1}_\I$ and $M \in \Iinterp{T\delta_2}_\I$; by
	the above $M \in \Iinterp{T\delta_1}_\I$ so that $M \Bind V \in \Iinterp{\tau_1}_\I$. By induction $\Iinterp{\tau_1}_\I \subseteq  \Iinterp{\tau_2}_\I$ 
	hence $M \Bind V \in \Iinterp{\tau_2}_\I$ so that $V \in \Iinterp{\delta_2 \to \tau_2}_\I$ by the choice of $M$.
\end{description} 
\QED

\begin{corollary}\label{cor:leq-Iinterp}
$T\omega_\Val$ is non trivial, and for all $\tau\in \ComType$ it is non trivial if and only if $\tau \leq T\omega_\Val$.
\end{corollary}

\Proof Where $T\omega_\Val =_\Comp \omega_\Comp$, by Lemma \ref{lem:leq-Iinterp} we would have $\Iinterp{T\omega_\Val}_\I = \Iinterp{\omega_\Comp}_\I$,
for any $\I$. But $\Iinterp{T\omega_\Val}_\I = \Set{M \in \ComTerm^0 \mid M \!\!\Downarrow} \neq \ComTerm^0 = \Iinterp{\omega_\Comp}_\I$ since
$\Omega_\Comp\!\not\Downarrow$.
The remaining part of the thesis now follows from Lemma \ref{lem:tau-neq-omega}.
\QED

We are now in place to show the only if part of Theorem \ref{thr:Conv-nonTrivial}.

\begin{lemma}\label{lem:Conv-nonTrivial}
$M\!\Downarrow~\Then~  \exists \,\mbox{non trivial}\ \tau .\, \der M:\tau$
\end{lemma}

\Proof If $M\!\Downarrow$ then $M \RedStar V$ for some $V \in \ValTerm^0$ by Lemma \ref{lem:ConvRed}; now $\der V:\omega_\Val$ so that
$\der \Unit V: T\omega_\Val$ by rule $\IntrUnit$; it follows that $\der M:T\omega_\Val$ by  
Corollary \ref{cor:subjectExpansion}, where $T\omega_\Val$ is non trivial by Corollary \ref{cor:leq-Iinterp}.
\QED

We say that a subset $X \subseteq \ComTerm^0$ is {\em saturated} if for all $M \in \ComTerm^0$, $M \Red N$ and $N \in X$ imply $M \in X$.

\begin{lemma}\label{lem:saturated}
For all $\tau \in \ComType$ and $\I$ the set $\Iinterp{\tau}_\I$ is saturated.
\end{lemma}

\Proof By induction over $\tau$. The case $\tau \equiv \omega_\Comp$ is trivial; the case $\tau \equiv \tau_1 \Inter \tau_2$ is immediate by induction.
Let $\tau \equiv T\delta$: then by hypothesis there exists  $V \in \Iinterp{\delta}_\I$ such that $N \!\!\Downarrow V$. By Lemma \ref{lem:ConvRed} we have
that $N \RedStar V$ so that $M \RedStar V$ and we conclude by the same lemma.
\QED

\begin{lemma}\label{lem:adequacy}
Let $\Gamma \der M:\tau$ where $\Gamma = \Set{x_1:\delta_1, \ldots , x_k:\delta_k}$ and $M\in \ComTerm$. For any $V_1, \ldots , V_k \in \ValTerm^0$ and $\I$, 
if $V_i \in \Iinterp{\delta_i}_\I$ for all $i = 1, \ldots , k$ then $M\Subst{V_1}{x_1}\cdots\Subst{V_k}{x_k} \in \Iinterp{\tau}_\I$.
\end{lemma}

\Proof We strength the thesis by adding that if $\Gamma \der W:\delta$ for $W \in \ValTerm$, then $W\Subst{V_1}{x_1}\cdots\Subst{V_k}{x_k} \in \Iinterp{\delta}_\I$ 
under the same hypotheses.
Then we reason by simultaneous induction over the derivations of $\Gamma \der M:\tau$ and $\Gamma \der W:\delta$. 
The cases of $(Ax)$ and $(\omega)$ are straightforward; cases $\IntrUnit$ and $\IntrInter$ are immediate by induction; case $(\leq)$ follows by induction
and Lemma \ref{lem:leq-Iinterp}. Let us abbreviate $M\Subst{\vec{V}}{\vec{x}} \equiv M\Subst{V_1}{x_1}\cdots\Subst{V_k}{x_k}$ and
similarly for $W\Subst{\vec{V}}{\vec{x}}$.

\begin{description}
\item Case $\IntrArr$: then the derivation ends by:
	\[\prooftree
		\Gamma, y:\delta' \der M':\tau'
	\justifies
	\Gamma \der \lambda y.M': \delta' \to \tau'
		\using  \IntrArr
	\endprooftree\]
	where $W \equiv  \lambda y.M'$ and $\delta \equiv \delta' \to \tau'$. Let $M'' \equiv M'\Subst{\vec{V}}{\vec{x}}$ and assume that
	$y\not\in \vec{x}$; to prove that $(\lambda y.M')\Subst{\vec{V}}{\vec{x}} \equiv \lambda y.M'' \in \Iinterp{\delta' \to \tau'}_\I$
	we have to show
	that $N \Bind \lambda y.M'' \in \Iinterp{\tau'}_\I$ for all $N \in \Iinterp{T\delta'}_\I$.
	
	Now if $N \in \Iinterp{T\delta'}_\I$ then there exists $V' \in \Iinterp{\delta'}_\I$ such that
	$N \Downarrow V'$. This implies that the hypothesis that $V_i \in \Iinterp{\delta_i}_\I$ for all $x_i:\delta_i \in \Gamma$ 
	now holds for the larger basis $\Gamma, y:\delta'$ so that by induction we have $M''\Subst{V'}{y} \in \Iinterp{\tau'}_I$. 	
	But
	\[N \Bind \lambda y.M'' \RedStar (\Unit V') \Bind \lambda y.M'' \Red M''\Subst{V'}{y}\]
	and the thesis follows since $\Iinterp{\tau'}_\I$ is saturated by Lemma \ref{lem:saturated}.

\item Case $\ElimArr$: then the derivation ends by:
	\[\prooftree
		\Gamma \der M': T\delta \quad \Gamma \der W':\delta\to \tau
	\justifies
		\Gamma \der M' \Bind W': \tau 
	\endprooftree\]
	where $M \equiv M' \Bind W'$. Let $M'' \equiv M'\Subst{\vec{V}}{\vec{x}}$ and $W''  \equiv W'\Subst{\vec{V}}{\vec{x}}$, so that
	$(M' \Bind W')\Subst{\vec{V}}{\vec{x}} \equiv M'' \Bind W''$. By induction $M'' \in \Iinterp{T\delta}_\I$
	and $W'' \in  \Iinterp{\delta\to \tau}_\I$ and the thesis follows by definition of the set $\Iinterp{\delta\to \tau}_\I$.
\end{description}
\QED

\noindent {\bf Proof of Theorem \ref{thr:Conv-nonTrivial}.} By Lemma \ref{lem:Conv-nonTrivial} it remains to show that if $\der M:\tau$ for
some non trivial $\tau$ then $M\!\!\Downarrow$. Since $M\in \ComTerm^0$ and the basis is empty, the hypothesis of Lemma \ref{lem:adequacy}
are vacuously true, so that we have $M \in \Iinterp{\tau}_\I$ for all $\I$. On the other hand, by Corollary \ref{cor:leq-Iinterp}, 
we know that $\tau \leq_\Comp T\omega_\Val$ since $\tau$ is non trivial. By Lemma \ref{lem:leq-Iinterp} it follows that
$M \in \Iinterp{\tau}_\I \subseteq \Iinterp{T\omega_\Val}_\I = \Set{N \in \ComTerm^0 \mid N\!\!\Downarrow}$ and we conclude.
\QED

In the next corollary we write $\Sem{M}^{TD_*}$ for $M \in \ComTerm^0$ omitting the term environment $\rho$ which is irrelevant.

\begin{corollary}[Computational Adequacy]
In the model $D_*$ we have that for any $M \in \ComTerm^0$:
\[M\!\Downarrow~\Iff~ \Sem{M}^{TD_*} \neq \bot_{TD_*}\]
\end{corollary}

\Proof By Lemma \ref{lem:type-semantics}, $\Sem{M}^{TD_*}  = \Set{\tau \in \ComType \mid \;\der M:\tau}$. By Theorem \ref{thr:Conv-nonTrivial}
$M\!\Downarrow$ if and only if $\der M:\tau$ for some non trivial $\tau$; but
\[\bot_{TD_*} = \!\!\Up \omega_\Comp \subset \!\!\Up\tau \subseteq  \Sem{M}^{TD_*}\]
where the inclusion $\Up \omega_\Comp \subset \!\!\Up\tau$ is strict since $\tau$ is non trivial.
\QED


\newcommand{\MoggilambdaComp}{\lambda_{\Comp}}
\newcommand{\MoggiRed}{>}
\newcommand{\Betav}{\beta_v}
\newcommand{\IdRed}{\textit{id}}
\newcommand{\CompRed}{\textit{comp}}
\newcommand{\Trad}[1]{\ulcorner #1\urcorner}
\newcommand{\MvTrad}[1]{\llcorner #1\lrcorner^{\Val}}
\newcommand{\McTrad}[1]{\llcorner #1\lrcorner^{\Comp}}
\newcommand{\MTrad}[1]{\llcorner #1\lrcorner}

\section{The untyped computational $\lambda$-calculus: $\TFlambdaComp$ vs. Moggi's $\MoggilambdaComp$ calculus}\label{sec:MoggiTrasl}

\begin{definition}[Values and computations]\label{def:MoggiTerms}	
	The {\em Moggi's computational $\lambda$-calculus}, shortly $\MoggilambdaComp$, is a calculus of two sorts of expressions:
	\[\begin{array}{r@{\hspace{0.7cm}}rll@{\hspace{1cm}}l}
	
	\mbox{Terms}: & e, e' & ::= & v \mid n & \\ [1mm]
	\mbox{Values}: & v & ::= & x \mid \lambda x.e & \\ [1mm]
	\mbox{NonValues}: & n & ::= & \Let{x}{e}{e'} \mid ee' & \mbox{(computations)}
	\end{array}\]

\end{definition}

\begin{definition}[Reduction]\label{def:MoggiReduction}
	The reduction relation $\MoggiRed\; \subseteq\,\Term \times \Term$ is defined as follows:
	\[\begin{array}{l@{\hspace{0.4cm}}rll@{\hspace{0.4cm}}l}
	\Betav & (\lambda x.e)v & \MoggiRed & e\Subst{v}{x} \\ [1mm]
	\eta_v &  \lambda x. vx & \MoggiRed & v \\ [1mm]
	\IdRed &  \Let{x}{e}{x} & \MoggiRed & e \\ [1mm]
	\CompRed &  (\Let{x_2}{(\Let{x_1}{e_1}{e_2})}{e}) & \MoggiRed &  \\ [1mm]
	&  (\Let{x_1}{e_1}{(\Let{x_2}{e_2}{e})}) &  &  \\ [1mm]
	let_v &  (\Let{x}{v}{e}) & \MoggiRed & e\Subst{v}{x} \\ [1mm]
	let_{.1}&  ne & \MoggiRed & (\Let{x}{n}{xe}) \\ [1mm]
	let_{.2} &  vn & \MoggiRed & (\Let{x}{n}{vx} \\ [1mm]
	\end{array}\]
	where $e\Subst{v}{x}$ denotes the capture avoiding substitution of $v$ for all free occurrences of $x$ in $e$. 
\end{definition}

\begin{definition}[Compatible Closure]
	Let $M,N,P\in \Term$ and suppose $M\MoggiRed N$:
	\[\begin{array}{l@{\hspace{0.4cm}}rll@{\hspace{0.4cm}}l}
	 & \lambda x.M & \MoggiRed & \lambda x.N \\ [1mm]
	 & MP\MoggiRed NP &  & PM\MoggiRed PN \\ [1mm]
	 & \Let{x}{M}{P} & \MoggiRed & \Let{x}{N}{P} \\ [1mm]
	 & \Let{x}{P}{M} & \MoggiRed & \Let{x}{P}{N} \\ [1mm]
	\end{array}\]
\end{definition}

\subsection{Traslation of $\TFlambdaComp$ into $\MoggilambdaComp$}
Define a function from $(\TFlambdaComp)$-$\Term$ to $(\MoggilambdaComp)$ -$\Term$ as follows:\\
$\Trad{\cdot}:(\TFlambdaComp)-\Term \rightarrow (\MoggilambdaComp) -\Term$
\[\begin{array}{l@{\hspace{0.4cm}}rll@{\hspace{0.4cm}}l}
\Trad{x}\equiv x &  & \Trad{\lambda x.M}\equiv \lambda x.\Trad{M}  &  \\ [1mm]
\Trad{\Unit V}\equiv \Trad{V}&  &\Trad{M\Bind V}\equiv \Let{x}{\Trad{M}}{\Trad{V}x} \mbox{ (for $x$ fresh)}  & 
\end{array}\]

\begin{lemma}[Substitution lemma for $\Trad{\cdot}$]\label{SubLemTrad}
	Let $M,V,W \in (\TFlambdaComp)$-$\Term$
	\begin{enumerate}[label=(\roman*)]
		\item $\Trad{V\Subst{W}{x}}\equiv \Trad{V}\Subst{\Trad{W}}{x}$
		\item $\Trad{M\Subst{W}{x}}\equiv \Trad{M}\Subst{\Trad{W}}{x}$
	\end{enumerate}
\end{lemma}
\Proof
By induction on the complexity of $(\TFlambdaComp)$-$\Term$:

	\[\begin{array}{l@{\hspace{0.4cm}}rll@{\hspace{0.4cm}}l}
& \Trad{x\Subst{W}{x}} & \equiv & \Trad{W}\equiv \Trad{x}\Subst{\Trad{W}}{x} \\ [2mm]
& \Trad{x\Subst{W}{x}} & \equiv & y\equiv \Trad{y}\Subst{\Trad{W}}{x} & \text{ where } y\not\equiv x \\ [2mm]
& \Trad{(\lambda z.M)\Subst{W}{x}} & \equiv & \Trad{\lambda z.M\Subst{W}{x}} \\ [1mm]
&  & \equiv & \lambda z.\Trad{M\Subst{W}{x}} \\ [1mm]
&  & \equiv & \lambda z.\Trad{M}\Subst{\Trad{W}}{x} & by i.h. \\ [1mm]
&  & \equiv & \Trad{\lambda z.M}\Subst{\Trad{W}}{x}
\end{array}\]

	\[\begin{array}{l@{\hspace{0.4cm}}rll@{\hspace{0.4cm}}l}
& \Trad{(\Unit V)\Subst{W}{x}} & \equiv & \Trad{\Unit V\Subst{W}{x}} \\ [1mm]
&  & \equiv & \Trad{V\Subst{W}{x}} & \text{ by i.h.} \\ [1mm]
&  & \equiv & \Trad{V}\Subst{\Trad{W}}{x} \\ [1mm]
&  & \equiv & \Trad{\Unit V}\Subst{\Trad{W}}{x} \\ [2mm]
& \Trad{(M\Bind V)\Subst{W}{x}} & \equiv & \Trad{M\Subst{W}{x}\Bind V\Subst{W}{x}}  \\ [1mm]
&  & \equiv & \Let{z}{\Trad{M\Subst{W}{x}}}{\Trad{V\Subst{W}{x}}z} \\ [1mm]
&  & \equiv & \Let{z}{\Trad{M}\Subst{\Trad{W}}{x}}{\Trad{V}\Subst{\Trad{W}}{x}z} & \mbox{ by i.h.}\\ [1mm]
&  & \equiv & (\Let{z}{\Trad{M}}{\Trad{V}z})\Subst{\Trad{W}}{x} \\ [1mm]
&  & \equiv & (\Trad{M\Bind V})\Subst{\Trad{W}}{x}
\end{array}\]
\QED

\begin{lemma}
	Let $M,N\in (\TFlambdaComp)$-$\Term$. If $M\Red N$, then $\Trad{M}=\Trad{N}$, where $=$ stands for the convertibility relation induced by $\MoggiRed$.
\end{lemma}
\Proof
Proof by induction on the generation of $M\Red N$.
\begin{description}
	\item[$\Betac$] $(\Unit V)\Bind (\lambda x. M) \Red M\Subst{V}{x}$\\
	\[\begin{array}{l@{\hspace{0.4cm}}rll@{\hspace{0.4cm}}l}
	&\Trad{(\Unit V)\Bind (\lambda x. M)}&\equiv& \Let{z}{\Trad{\Unit V}}{\Trad{(\lambda x.M)}z}\\[1mm]
	& & \equiv & \Let{z}{\Trad{ V}}{(\lambda x.\Trad{M})z}
	\end{array}\]
	Since $\Trad{V}\in \textit{Values}$, one has: 
	\[\begin{array}{lrlll}
	\Let{z}{\Trad{ V}}{(\lambda x.\Trad{M})z}&\MoggiRed & (\lambda x. \Trad{M})\Trad{V} & \mbox{ by \textit{$let_v$} and } z\not\in \FV (\Trad{M})\\[1mm]
	 & \MoggiRed & \Trad{M}\Subst{\Trad{V}}{x} &\mbox{ by $\Betav$}\\[1mm]
	 &\equiv & \Trad{M\Subst{V}{x}}& \mbox{ by substitution lemma \ref{SubLemTrad}}
	\end{array}\]
	
	\item[$(id)$] $M\Bind \lambda x. \Unit x \Red M$
	\[\begin{array}{l@{\hspace{0.4cm}}rll@{\hspace{0.4cm}}l}
	&\Trad{M\Bind \lambda x. \Unit x}&\equiv& \Let{z}{\Trad{M}}{\Trad{(\lambda x.\Unit x)}z}\\[1mm]
	& & \equiv & \Let{z}{\Trad{M}}{(\lambda x.x)z}\\[1mm]
	& &\MoggiRed &  \Let{z}{\Trad{M}}{z} \MoggiRed \Trad{M} \mbox{ by $\MoggilambdaComp$-\IdRed}
	\end{array}\]
	
	\item[$(ass)$] $(L\Bind \lambda x. M)\Bind (\lambda y. N)\Red L\Bind \lambda x. (M\Bind \lambda y. N)$ where $x\not\in\FV (N)$
	\[\begin{array}{lrlll}
	&\Trad{(L\Bind \lambda x. M)\Bind (\lambda y. N)}&\equiv& \Let{z_2}{\Trad{(L\Bind \lambda x. M)}}{\Trad{(\lambda y. N)}z_2}\\[1mm]
	& & \equiv & \Let{z_2}{\Let{z_1}{\Trad{L}}{\Trad{\lambda x. M}z_1}}{\Trad{(\lambda y. N)}z_2}\\[1mm]
	& & \MoggiRed & \Let{z_1}{\Trad{L}}{(\Let{z_2}{(\lambda x. \Trad{M}z_1)}{\Trad{(\lambda y.N)}z_2})} & \mbox{ by $\MoggilambdaComp$-comp}\\[1mm]
	& & \equiv & \Let{z_1}{\Trad{L}}{(\Let{z_2}{(\lambda x.\Trad{M})z_1}{\Trad{(\lambda y.N)}z_2})}\\[1mm]
	& & \equiv & \Let{z_1}{\Trad{L}}{(\Let{z_2}{\Trad{M\Subst{z_1}{x}}}{\Trad{(\lambda y.N)}z_2})}\\[1mm]
	& & < & \Let{z_1}{\Trad{L}}{(\lambda x.\Let{z_2}{\Trad{M}}{\Trad{(\lambda y.N)}z_2})z_1}\\[1mm]
	& & \equiv & \Trad{L \Bind \lambda x. (M\Bind \lambda y.N)}
	\end{array}\]
	
	Since $z\not\equiv z_2$, $x\not\in\FV (N)$ and then $x\not\in\FV (\Trad{N})$. From this consideration, one has:\\
	$(\Let{z_2}{\Trad{M}}{(\Trad{\lambda y.N})z_2})\Subst{z_1}{x} \equiv \Let{z_2}{\Trad{M\Subst{z_2}{x}}}{\Trad{\lambda y.N}z_2}$
\end{description} 
\QED

\begin{remark}
	Reduction $\Red$ over $(\TFlambdaComp)$-$\Term$ has been defined devoid of $\eta$ nor $\xi$ rules. Both can be added by extending reduction relation over values:
	\[\begin{array}{l@{\hspace{0.4cm}}rll@{\hspace{0.4cm}}l}
	\xi: & 
	\prooftree
	M \Red N
	\justifies
	\lambda x.M \Red \lambda x.N
	\endprooftree
	& \qquad \eta: &\lambda x. \Unit x\Bind V\Red V &\mbox{ if } x\not\in \FV (V)
	\end{array}\]
	With respect to previous lemma, concerning $\xi$-rule: $\Trad{M}=\Trad{N}$ by induction hypothesis and $\Trad{\lambda x.M}=\lambda x.\Trad{M}=\lambda x.\Trad{N}=\Trad{\lambda x.N}$.\\
	$\eta$:
	\[\begin{array}{l@{\hspace{0.4cm}}rll@{\hspace{0.4cm}}l}
	&\Trad{\lambda x. \Unit x\Bind V}&\equiv &\lambda x. \Let{z}{x}{\Trad{V}z}\\[1mm]
	& & \MoggiRed & \lambda x. \Trad{V}z &\mbox{ by $let_v$}\\[1mm]
	& & \MoggiRed & \Trad{V} &\mbox{ by $eta_c$}
	\end{array}\]
\end{remark}

\subsection{Traslation of $\MoggilambdaComp$ into  $\TFlambdaComp$ }

Define a function from $(\MoggilambdaComp)$ -$\Term$ to $(\TFlambdaComp)$-$\Term$ as follows:\\
$\MvTrad{\cdot}: (\MoggilambdaComp) -Values \rightarrow \ValTerm$
$\qquad \McTrad{\cdot}: (\MoggilambdaComp) -NonValues \rightarrow \ComTerm$

\[\begin{array}{l@{\hspace{0.4cm}}rll@{\hspace{0.4cm}}l}
& \MvTrad{x}= x &  \\ [1mm]
\MvTrad{\lambda x.v}= \lambda x.\Unit \MvTrad{v}  &  & \MvTrad{\lambda x.n}= \lambda x.\McTrad{n}   
\end{array}\]

\[\begin{array}{l@{\hspace{0.2cm}}rll@{\hspace{0.4cm}}l}
\McTrad{vn} =\McTrad{n} \Bind \MvTrad{v}   \\ [1mm]
\McTrad{v v'}= \Unit \MvTrad{v'}\Bind \MvTrad{v}\\[1mm]
\McTrad{nv}= \McTrad{n}\Bind (\lambda x.\Unit\MvTrad{v}\Bind x) & \mbox{ for fresh }x\\[1mm]
\McTrad{n n'}= \McTrad{n}\Bind (\lambda x.\McTrad{n'}\Bind x) & \mbox{ for fresh }x\\[1mm]
\McTrad{\Let{x}{n}{n'}}=\McTrad{n}\Bind \lambda x.\McTrad{n'} \\[1mm]
\McTrad{\Let{x}{v}{n'}}=\Unit\MvTrad{v}\Bind \lambda x.\McTrad{n'} \\[1mm]
\McTrad{\Let{x}{n}{v}}= \McTrad{n}\Bind \lambda x.\Unit\MvTrad{v} \\[1mm]
\McTrad{\Let{x}{v}{v'}}=\Unit\MvTrad{v}\Bind \lambda x.\Unit\MvTrad{v'}
\end{array}\]

\begin{lemma}[Substitution lemma for $\MTrad{\cdot}$]\label{SubLemMTrad}
	Let $w,v \in Values$ and $n\in NonValues$;
	\begin{enumerate}[label=(\roman*)]
		\item $\MvTrad{w\Subst{v}{x}}\equiv \MvTrad{w}\Subst{\MvTrad{v}}{x}$
		\item $\McTrad{n\Subst{Wv}{x}}\equiv \McTrad{n}\Subst{\MvTrad{v}}{x}$
	\end{enumerate}
\end{lemma}
\Proof

\textit{(i)}

$w\equiv x:$
\[\begin{array}{l}
\MvTrad{x\Subst{v}{x}}=\MvTrad{v}=\MvTrad{x}\Subst{\MvTrad{v}}{x}
\end{array}\]

$w\equiv y\not\equiv x:$
\[\begin{array}{l}
\MvTrad{y\Subst{v}{x}}=\MvTrad{y}=y=\MvTrad{y}\Subst{\MvTrad{v}}{x}
\end{array}\]

$w\equiv \lambda y.w$, where $y\not\equiv x$ and $y\not\in \FV (w)$:
\[\begin{array}{lll}
\MvTrad{\lambda y.w\Subst{v}{x}}&=\MvTrad{\lambda y.(w\Subst{v}{x})}\\[1mm]
&=\lambda y. \Unit \MvTrad{w\Subst{v}{x}}\\[1mm]
&=\lambda y. \Unit (\MvTrad{w}\Subst{\MvTrad{v}}{x}) & \mbox{ by i.h.}\\[1mm]
&=(\lambda y. \Unit \MvTrad{w})\Subst{\MvTrad{v}}{x} \\[1mm]
&=\MvTrad{\lambda y. w}\Subst{\MvTrad{v}}{x}
\end{array}\]

$w\equiv \lambda y.n:$

\[\begin{array}{lll}
\MvTrad{(\lambda y.n)\Subst{v}{x}}&=\MvTrad{\lambda y.(n\Subst{v}{x})}\\[1mm]
&=\lambda y. \McTrad{n\Subst{v}{x}}	\\[1mm]
&=\lambda y.\McTrad{n}\Subst{\MvTrad{v}}{x}& \mbox{ by i.h.}\\[1mm]
&=(\lambda y.\McTrad{n})\Subst{\MvTrad{v}}{x}\\[1mm]
&=\MvTrad(\lambda y.{n)}\Subst{\MvTrad{v}}{x}
\end{array}\]

\textit{(ii)}

$n\equiv wn':$
\[\begin{array}{lll}
\McTrad{(wn')\Subst{v}{x}}&=\McTrad{(w\Subst{v}{x})(n'\Subst{v}{x})}\\[1mm]
&=\McTrad{n'\Subst{v}{x}}\Bind \MvTrad{w\Sub{v}{n}}\\[1mm]
&=\McTrad{n'}\Subst{\MvTrad{v}}{x}\Bind \MvTrad{w}\Sub{\MvTrad{v}}{x}& \mbox{ by i.h.}\\[1mm]
&=(\McTrad{n'}\Bind  \MvTrad{w})\Sub{\MvTrad{v}}{x}\\[1mm]
&= \McTrad{wn'}\Sub{\MvTrad{v}}{x}
\end{array}\]

$n\equiv ww'$
\[\begin{array}{lll}
\MvTrad{(ww')\Subst{v}{x}}&=\MvTrad{(w\Subst{v}{x})(w'\Subst{v}{x})}\\[1mm]
&=\Unit \MvTrad{w'\Subst{v}{x}}\Bind \MvTrad{w\Sub{v}{n}}\\[1mm]
&=\Unit \MvTrad{w'}\Subst{\MvTrad{v}}{x}\Bind \MvTrad{w}\Sub{\MvTrad{v}}{x}& \mbox{ by i.h.}\\[1mm]
&=(\MvTrad{w'}\Bind  \MvTrad{w})\Sub{\MvTrad{v}}{x}\\[1mm]
&= \McTrad{ww'}\Sub{\MvTrad{v}}{x}
\end{array}\]

$n\equiv n'w$
\[\begin{array}{lll}
\McTrad{(n'w)\Subst{v}{x}}&=\McTrad{(n'\Subst{v}{x})(w\Subst{v}{x})}\\[1mm]
&= \McTrad{n'\Subst{v}{x}}\Bind (\lambda y.\Unit \MvTrad{w\Sub{v}{n}} \Bind y)\\[1mm]
&= \McTrad{n'}\Subst{\MvTrad{v}}{x}\Bind (\lambda y.\Unit \MvTrad{w}\Sub{\MvTrad{v}}{n} \Bind y)\\[1mm]
&= (\McTrad{n'}\Bind (\lambda y.\Unit \MvTrad{w}\Bind y))\Subst{\MvTrad{v}}{x}\\[1mm]
&=\McTrad{n'w}\Subst{\MvTrad{v}}{x}
\end{array}\]

$n\equiv n'm$
\[\begin{array}{lll}
\McTrad{(n'm)\Subst{v}{x}}&=\McTrad{(n'\Subst{v}{x})(m\Subst{v}{x})}\\[1mm]
&= \McTrad{n'\Subst{v}{x}}\Bind (\lambda y.\McTrad{m\Sub{v}{n}} \Bind y)\\[1mm]
&= \McTrad{n'}\Subst{\MvTrad{v}}{x}\Bind (\lambda y.\McTrad{m}\Sub{\MvTrad{v}}{n} \Bind y)\\[1mm]
&= (\McTrad{n'}\Bind (\lambda y.\McTrad{m} \Bind y)\Sub{\MvTrad{v}}{n} \\[1mm]
&=\McTrad{n'm}{\MvTrad{v}}{n}
\end{array}\]

$n\equiv (\Let{y}{m}{n'})$
\[\begin{array}{lll}
\McTrad{(\Let{y}{m}{n'})\Subst{v}{x}}&=\McTrad{(\Let{y}{m\Subst{v}{x}}{n'\Subst{v}{x}})}\\[1mm]
&=\McTrad{m\Subst{v}{x}}\Bind \lambda y.\McTrad{n'\Subst{v}{x}}\\[1mm]
&=\McTrad{m}\Subst{\MvTrad{v}}{x}\Bind \lambda y.\McTrad{n'}\Subst{\MvTrad{v}}{x}& \mbox{ by i.h.}\\[1mm]
&=(\McTrad{m}\Bind \lambda y.\McTrad{n'})\Subst{\MvTrad{v}}{x}\\[1mm]
&=\McTrad{\Let{y}{m}{n'}}\Subst{\MvTrad{v}}{x}
\end{array}\]

Similarly, the statement is proved for remaining cases.
\QED

\begin{definition}\label{def:MTrad}
	Let's define a function $\MTrad{\cdot}$ that maps $\MoggilambdaComp$-$\Term$ into $(\TFlambdaComp)$-$\ComTerm$, as follows:\\
	For all $n\in NonValues$ and $v\in Values$
	\[\begin{array}{ll}
	\MTrad{n}=\McTrad{n} & \MTrad{v}=\MvTrad{v}
	\end{array}\]
\end{definition}

\begin{corollary}
	$\MTrad{e\Subst{v}{x}}=\MTrad{e}\Subst{\MvTrad{v}}{x}$
\end{corollary}
\Proof By induction on the structure of $e\in\MoggilambdaComp$-$\Term$:\\
$e\equiv w\in Values$, then $\MTrad{w\Subst{v}{x}}=\Unit \MvTrad{w\Subst{v}{x}}=\Unit (\MvTrad{w}\Subst{\MvTrad{v}}{x})$by part \textit{(i)} of \ref{SubLemMTrad}. This is equivalent to $(\Unit \MvTrad{w})\Subst{\MvTrad{v}}{x}$, that is equivalent to $\MTrad{w}\Subst{\MvTrad{v}}{x}$ hence the thesis by definition \ref{def:MTrad}.\\
If $e\equiv n\in NonValues$:
\[\begin{array}{lll}
\MTrad{n\Subst{v}{x}}&=\McTrad{n\Subst{v}{x}}\\[1mm]
&=\McTrad{n}\Subst{\MvTrad{v}}{x} & \mbox{ by part \textit{(ii)} of \ref{SubLemMTrad}}\\[1mm]
&=\MTrad{n}\Subst{\MvTrad{v}}{x}
\end{array}\]
\QED

\begin{lemma}
	$e\MoggiRed e' \Rightarrow \MTrad{e}\xrightarrow{*}\MTrad{e'}$
\end{lemma}
\Proof
By induction on the definition of $\MoggiRed$:
\[\begin{array}{lll}
\MTrad{(\lambda x.w)v}&=\McTrad{(\lambda x.w)v}\\[1mm]
=\Unit \MvTrad{v}\Bind \lambda x.\Unit\MvTrad{w}&\Betac (\Unit \MvTrad{w})\Subst{\MvTrad{v}}{x}\\[1mm]
&=\MTrad{w} \Subst{\MvTrad{v}}{x} & \mbox{ by definition \ref{def:MTrad}}\\[1mm]
&=\MTrad{w\Subst{v}{x}}& \mbox{ by the previous corollary}
\end{array}\]

\[\begin{array}{lll}
\MTrad{(\lambda x.n)v}&=\McTrad{(\lambda x.n)v}\\[1mm]
=\Unit \MvTrad{v}\Bind \lambda x.\McTrad{n}&\Betac \McTrad{n}\Subst{\MvTrad{v}}{x}\\[1mm]
&=\McTrad{n\Subst{\MvTrad{v}}{x}} & \mbox{ by part \textit{(ii)} of \ref{SubLemMTrad}}\\[1mm]
&=\MTrad{n\Subst{v}{x}}& \mbox{ by definition \ref{def:MTrad}}
\end{array}\]

$\eta_C$: $\MTrad{\lambda x.vx}=\lambda x. \McTrad{vx} =\lambda x. \Unit x \Bind \MvTrad{v}\Red_{\eta_C}\MvTrad{v}$\\
\textit{(id).1}: $(\Let{x}{v}{x})\MoggiRed v$
\[\begin{array}{lll}
\MTrad{\Let{x}{v}{x}}&=\McTrad{\Let{x}{v}{x}}\\[1mm]
&=\Unit \MvTrad{v}\Bind \lambda x.\Unit \MvTrad{x}\\[1mm]
&=\Unit \MvTrad{v}\Bind \lambda x.\Unit x\\[1mm]
\Red_{(id)}\Unit \MvTrad{v}=\MvTrad{v}
\end{array}\]

Regarding rule \textit{(comp)}, the proof is done after proving 8 different cases, here we will show just two of them. 
\[\begin{array}{lll}
\MTrad{\Let{y}{\Let{x}{m}{n}}{n'}}&= \McTrad{\Let{y}{\Let{x}{m}{n}}{n'}}\\[1mm]
&=(\McTrad{m}\Bind \lambda x.\McTrad{n})\Bind \lambda y.\McTrad{n'}\\[1mm]
&\Red_{(ass)} \McTrad{m}\Bind \lambda x. (\McTrad{n}\Bind \lambda y. \McTrad{n'})\\[1mm]
&=\MTrad{\Let{x}{m}{\Let{y}{n}{n'}}}
\end{array}\]

\[\begin{array}{lll}
\MTrad{\Let{y}{\Let{x}{v}{n}}{n'}}&= (\Unit\MvTrad{v}\Bind\lambda x.\McTrad{n})\Bind \lambda y.\McTrad{n'}\\[1mm]
&\Red_{(ass)} \Unit\MvTrad{v}\Bind\lambda x.(\McTrad{n})\Bind \lambda y.\McTrad{n'})\\[1mm]
&=\Let{x}{\Unit\MvTrad{v}}{\Let{y}{\McTrad{n}}{\McTrad{n'}}}\\[1mm]
&=\MTrad{\Let{x}{v}{\Let{y}{n}{n'}}}
\end{array}\]

$let_v$ case splits in two different sub-cases:
\[\begin{array}{lll}
\McTrad{\Let{x}{v}{n}}&=\Unit \MvTrad{v}\Bind\lambda x.\McTrad{n}\\[1mm]
&\Red_{\beta_v} \McTrad{n}\Subst{\MvTrad{v}}{n}\\[1mm]
&=\MTrad{n\Subst{v}{x}}& \mbox{  by part \textit{(ii)} of \ref{SubLemMTrad}}
\end{array}\]

\[\begin{array}{lll}
\McTrad{\Let{x}{v}{w}}&=\Unit \MvTrad{v}\Bind\lambda x.\Unit \MvTrad{w}\\[1mm]
&\Red_{\beta_v} (\Unit \MvTrad{w})\Subst{\MvTrad{v}}{n}\\[1mm]
&=\MTrad{w}\Subst{\MvTrad{v}}{n}& \mbox{ by definition of }\MvTrad{\cdot}\\[1mm]
&=\MTrad{w\Subst{v}{x}}& \mbox{  by the previous corollary}
\end{array}\]

$let_1$ case splits in two different sub-cases:
\[\begin{array}{lll}
\McTrad{nv}=\McTrad{n}\Bind \lambda x.\Unit \MvTrad{v}\Bind x\\[1mm]
\McTrad{\Let{x}{n}{xv}}=\McTrad{n}\Bind\lambda x. \McTrad{xv}=\McTrad{n}\Bind(\lambda x. \Unit\MvTrad{v}\Bind x)
\end{array}\]
Since the above equalities, one has proved that $\McTrad{nv}= \McTrad{\Let{x}{n}{xv}}$. Then, a fortiori $\McTrad{nv}\xrightarrow{*} \McTrad{\Let{x}{n}{xv}}$, since $\xrightarrow{*}$ is also a reflexive closure of $\Red$.

\[\begin{array}{lll}
\McTrad{nm}=\McTrad{n}\Bind \lambda x.\McTrad{m}\Bind x\\[1mm]
\McTrad{\Let{x}{n}{xm}}=\McTrad{n}\Bind\lambda x. \McTrad{xm}=\McTrad{n}\Bind(\lambda x. \McTrad{m}\Bind x)
\end{array}\]
The conclusion follows analogous considerations stated in the previous subcase.  
\QED

\begin{theorem}
	There exists an interpretation $\MTrad{\cdot}$ from $\MoggilambdaComp$ into $\TFlambdaComp$ that preserves reductions.\\
	There exists an interpretation $\Trad{\cdot}$ from   $\TFlambdaComp$ into $\MoggilambdaComp$ that preserves the convertibility relation.

\end{theorem}

\section{Related and future works}\label{sec:related}
The main inspiration for the present work has been \cite{LagoGL17}, not because we develop the co-algebraic
approach to equality treated there, but since the considered calculus is essentially type free, and also
because of the idea of investigating reasoning principles that hold for any monad in general and its algebras,
originating from \cite{PlotkinP03} and related works.

Since Moggi's seminal papers \cite{Moggi'89,Moggi'91}, a substantial body of research has been carried out about
the computational $\lambda$-calculus and the concept of monad, both in theory and in practice of functional programming languages.
Here, because of the large bibliography on the subject, we shall refer to the closest related works.

The calculus $\TFlambdaComp$ in subsection \ref{subsec:reduction} is a subcalculus of Wadler's one in \cite{Wadler-Monads}, 
but it includes terms like a minimal version of self-application $\Unit x \Bind x$, that has no simple type. 
The untyped calculus which is closest to $\TFlambdaComp$ is Moggi's type free calculus in \cite{Moggi'89} \S 6, called $\lambda_c$. 
However, the latter is much richer than ours, as translation in section \ref{sec:MoggiTrasl} show.
The translation $\MTrad{\cdot}$ preserves reduction, but we have to add $\eta_c$ to $\MonRed$, hence loosing its confluence.

By this, both the confluence proof of $\lambda_c$ in \cite{MaraistOTW99} \S 8.3, where it is called {\sc COMP}, and the implicit proof
in \cite{Hamana18}, cannot be used in our case, and we prefer to prove confluence of $\MonRed$ from scratch, although following a
standard pattern: see e.g. \cite{AriolaFMOW95}. Confluence of  {\sc COMP} is established in \cite{MaraistOTW99} via a translation from a 
call-by-need linear $\lambda$-calculus, but
without $\eta$, facing a similar difficulty as we mention here at the end of subsection \ref{subsec:reduction}.

Intersection types have been used in
\cite{Davies-Pfenning'00} in a $\lambda$-calculus with side effects and reference types. In their work a problem appears, since left distributivity of the arrow over
intersection (a rule in \cite{BarendregtDS2013} as well as is an axiom of the theory $\Th_\Val$ in Definition \ref{def:type-theories-Th_V-Th_C} above) 
is unsound. This is remedied by restricting intersection introduction to values. However Davies and Pfenning's work is not concerned with monads,
so that value and non-value types are of the same sort. On the contrary by working in a system like in Definition \ref{def:typeAssignment} 
these types are distinct, and actual definitions of unit and bind for the state monad are reflected into their typings: it is then worthy to investigate
such a system, that should avoid unsound typings without ad hoc restrictions.

The convergence relation in section \ref{sec:adequacy} is the adaptation of a similar concept introduced in \cite{Abramsky'90,Abramski-Ong'93}, 
which represent the observable
property at the basis of of the definition of the applicative bisimulation. It shares some similarity with the convergence relation considered in \cite{LagoGL17},
where Abramsky's idea is extended to a computational $\lambda$-calculus very similar to $\TFlambdaComp$. However, differently than
in Abramsky's work and the relation in Definition \ref{def:convergence}, 
the authors define their predicate as a relation among syntax and semantics, co-inductively defining the interpretation of terms.

Theorem \ref{thr:Conv-nonTrivial}, characterizing convergent terms by non trivial typings, is evidence of the expressive power of our system.
However, since convergence is undecidable, typability in the system is also undecidable. If the system should be useful in practice, say as
a method for abstract interpretation and static analysis or program syntesis, then restricted subsystems should be considered, like bounded intersection type systems
recently proposed in \cite{DudderMRU12,DudenhefnerR17-a,DudenhefnerR17-b}.

Another research line is to investigate how inductive and co-inductive program properties can be
handled in our framework by exploiting the relational properties of domains and of the inverse limit construction studied in \cite{Pitts96}, that generalize the concept of admissible set we used to
define monadic type interpretations.

\section{Conclusion}\label{sec:conclusion}

Starting with the general domain equation of the type-free call-by-value computational $\lambda$-calculus
we have tailored term and type syntax and defined a type theory and an intersection type
assignment system that is sound and complete w.r.t. the interpretation of terms and types in a class of models.
This class is parametrically defined w.r.t. the monad at hand, as well as the system itself that induce
a family of filter models. If the monad is the trivial one, originating from the identity functor, then
our system collapses to BCD. The filter model we have obtained is for a generic monad; we conjecture that
the model is initial in a suitable category, to which all its instances belong. We leave this question to further
research.

\bibliographystyle{alpha}
\bibliography{references}

\appendix
\section{Proofs}\label{app:proofs}

\subsection{Models of $\lambda^u_c$}
\begin{lemma}\label{applem:subtsInterpretation}
	In any $T$-model $D$ we have $\Sem{M\Subst{V}{x}}^{TD}_\rho = \Sem{M}^D_{\rho[x\mapsto \Sem{V}^D_\rho]}$.
\end{lemma}
\Proof
	In order to prove the lemma, one should also prove the equivalent statement reformulate for values.
	\begin{center}
	$\Sem{W\Subst{V}{x}}^{D}_\rho = \Sem{W}^D_{\rho[x\mapsto \Sem{V}^D_\rho]}$
	\end{center}
	Then the proof easily follows by simultaneous induction on $W$ and $M$.
\QED

\begin{proposition}\label{appprop:soundness-of-interpretation}
	If $M \Red N$ then $\Sem{M}^{TD}_\rho = \Sem{N}^{TD}_\rho$ for any $T$-model $D$ and $\rho \in \Env_D$. Therefore,
	if $=$ is the convertibility relation of $\Red$, that is the symmetric closure of $\RedStar$, then $M = N$
	implies $\Sem{M}^{TD}_\rho = \Sem{N}^{TD}_\rho$.
\end{proposition}
\Proof
	By induction on the definition of $M\Red N$.\\
	In particular, case $\Betac: \Unit V \Bind (\lambda x.M) \rightarrow M\Subst{V}{x}$ follows by \ref{applem:subtsInterpretation}.
\QED

\subsection{Intersection type assignment system for $\TFlambdaComp$}

\begin{lemma}[Generation lemma] \label{app:genLemma}
	Assume that $\delta \neq \omega_\Val$ and $\tau \neq \omega_\Comp$, then:
	\begin{enumerate}
		\item $\Gamma \der x: \delta \Then \Gamma(x) \leq_\Val \delta$
		\item $\Gamma \der \lambda x.M: \delta \Then
		\exists I, \delta_i, \tau_i.~  \forall i \in I.\; \Gamma ,x:\delta_i \der M:\tau_i \And  \bigwedge_{i\in I}\delta_i \to \tau_i \leq_\Val \delta$
		\item $\Gamma \der \Unit V : \tau \Then \exists \delta.\; \Gamma \der V:\delta \And T\delta \leq_\Comp \tau$
		\item $\Gamma \der M\Bind V : \tau \Then$ \\
		$\exists I, \delta_i, \tau_i.~  \forall i \in I.\; \Gamma \der M:T\delta_i \And \Gamma\der V:\delta_i\to\tau_i \And
		\bigwedge_{i\in I}\tau_i \leq_\Comp \tau$
	\end{enumerate}
\end{lemma}
\Proof
	By induction on terms and their derivations. 
	\newline \newline
	\textit{1.} By induction on derivations. In this case only axiom \textit{(Ax)}, and rules $\IntrInter$, \textit{($\leq$)} could have been applied.\\
	In the first case $\Gamma (x)=\delta$.\\
	If $\IntrInter$ or \textit{($\leq$)} have been applied, the implication is proved by applying induction hypothesis.\newline \newline
	\textit{2.} By induction on derivations. The most interesting case is when $\delta = \delta_1 \Inter \delta_2$ and the last rule that is applied is $\IntrInter$, that is:
	\begin{center}
		$\prooftree
		\Gamma \der \lambda x.M:\delta_1 \quad \Gamma \der \lambda x.M:\delta_2
		\justifies
		\Gamma \der \lambda x.M:\delta_1 \Inter \delta_2
		\endprooftree$
	\end{center}
	by the induction hypothesis there exist $I',J'\subset \mathbb{N},\delta'_i, \tau'_i, \delta'_j, \tau'_j$ with $i\in I'$ and $j\in J'$, such that:
	\[
	\begin{array}{c}
		\Gamma , x:\delta'_i\der M:\tau'_i, \quad \bigwedge_{i\in I'} \delta'_i \to \tau'_i \leq \delta_1\\
		\Gamma , x:\delta'_j\der M:\tau'_j, \quad \bigwedge_{j\in J'} \delta'_j \to \tau'_j \leq \delta_2
	\end{array}
	\]
	The result follows, as $ (\bigwedge_{i\in I'} \delta'_i \to \tau'_i)\Inter (\bigwedge_{j\in J'} \delta'_j \to \tau'_j)\leq \delta$ and $I:=I'\cup J'$.
	\newline \newline
	\textit{3.} By induction on derivations.\\
	In this case possible applied rules are $\IntrInter$, $\IntrUnit$ or \textit{($\leq$)}.\\
	In the first case one has $\Gamma \der V:\delta$ and $\tau=T\delta$, the other two cases follows by the induction hypothesis. 
	\newline \newline
	\textit{4.} By induction on derivations.\\
	By assumption on $\tau$ and on the shape of $M\Bind V$, the final applied rules are $\ElimArr$, $\IntrInter$, \textit{($\leq$)}.\\
	In detail, the first case is
	\begin{center}
	$\prooftree
	\Gamma \der M:T\delta \quad \Gamma \der V:\delta \to \tau
	\justifies
	\Gamma \der M\Bind V:\tau
	\endprooftree $
	\end{center}
	One can take $I=\{ 1\}$, $\tau_1=\tau$, and $\delta_1=\delta$.\\
	In the second case, the final applied rule is \textit{($\leq$)} and the conclusion follows easily by the induction hypothesis.\\
	Finally, the third case is when $\IntrInter$ is the final rule applied. Thus
	\[
	\begin{array}{rcl}
	\tau=\tau_1 \Inter \tau_2 & \mbox{ and } & \prooftree
	\Gamma \der M\Bind V:\tau_1 \quad \Gamma \der M\Bind V:\tau_2
	\justifies
	\Gamma \der M\Bind V:\tau_1\Inter \tau_2
	\endprooftree
	\end{array}
	\]
	By the induction hypothesis there exist $I',J'\subset \mathbb{N},\delta'_i, \tau'_i, \delta'_j, \tau'_j$ with $i\in I'$ and $j\in J'$, such that:
	\[
	\begin{array}{rcl}
	\Gamma \der M:T\delta'_i, &\quad & \Gamma \der V : \delta'_i\to \tau'_i\\
	\Gamma \der M:T\delta'_j, &\quad &\Gamma \der V : \delta'_j\to \tau'_j\\
	\bigwedge_{i\in I'} \tau'_i\leq \tau_1 &\quad & \bigwedge_{j\in J'} \tau'_j\leq \tau_2
	\end{array}
	\]
	Hence the conclusion, as $ (\bigwedge_{i\in I'} \tau'_i )\Inter (\bigwedge_{j\in J'} \tau'_j)\leq \delta$ and $I:=I'\cup J'$.\\
\QED

\begin{lemma}[Substitution lemma] \label{app:SubstLemma} 
	If $\Gamma, x:\delta \der M:\tau$ and $\Gamma \der V:\delta$ then $\Gamma \der M\Subst{V}{x}: \tau$.
\end{lemma}
\Proof
	In order to prove the lemma, one should also strengthen it by adding its counterpart for values.
	\begin{center}
		$\Gamma, x:\delta \der W:\delta'$ and $\Gamma \der V:\delta \Rightarrow \Gamma \der W\Subst{V}{x}: \delta'$
	\end{center}
	Then the proof is by simultaneous induction on $W$ and $M$.\newline \newline
	We only give the proof for values:\\
	\textbf{Base case: } $W\equiv y$, trivial in both cases in which $x \not\equiv y$ or $x \equiv y$.\\
	\textbf{Inductive step: } $W\equiv \lambda y.M'$\\
	By Generation lemma \ref{app:genLemma}, if $\Gamma ,x:\delta \der \lambda y.M':\delta'$ then $\exists I, \delta'_i, \tau . \forall i\in I	\, \Gamma, x:\delta ,\delta'_i \der M:\tau_i \mbox{ and } \bigwedge_{i\in I}\	\delta_i \to \tau_i \leq \delta'$\\
	Thus 
	\[ \begin{array}{rl}
	\Gamma, y:\delta'_i \der M\Subst{V}{x} :\tau_i &\qquad \mbox{by IH}\\
	\Rightarrow \Gamma \der \lambda x.(M\Subst{V}{x}):\delta'_i \to \tau_i & \\
	\Rightarrow \Gamma \der (\lambda x. M)\Subst{V}{x}: \delta'_i \to \tau_i & \qquad \mbox{ by definition of substitution} \\
	\Rightarrow \Gamma \der (\lambda x. M)\Subst{V}{x}: \delta' & \qquad \mbox{ by } (\leq),\IntrInter \mbox{ rules }
	\end{array}\]
	Proof by induction on $M$ is quite specular.
\QED

\begin{theorem}[Subject reduction]\label{app:subjectReduction}
	If $\Gamma \der M:\tau$ and $M \Red N$ then $\Gamma \der N:\tau$.
\end{theorem}

\Proof
	By induction on derivations.\\
	Firstly, we consider the case $M\equiv \Unit V \Bind (\lambda x. M')$ and $N\equiv M'\Subst{V}{x}$, that is $M\Red_{\Betac}N$.\\
	Since $\Gamma \der M:\tau $, by Gen. Lemma \ref{app:genLemma}:
	\begin{equation*}
	\Rightarrow \exists I, \delta_i, \tau_i. \ \forall i\in I \ \Gamma \der \Unit V: T\delta_i \mbox{ and } \Gamma \der \lambda x.M':\delta_i \to \tau_i \mbox{ and } \bigwedge_{i\in I} \tau_i \leq \tau 
	\end{equation*}
	And by re applying twice Gen. Lemma
	\begin{equation*}
	\Rightarrow \exists I, \delta_i, \tau_i. \ \forall i\in I\   \exists \delta'_i\  \Gamma \der V: \delta'_i \mbox{ and } T\delta'_i\leq T\delta_i \mbox{ and }  \Gamma , x:\delta_i  \der M':\tau_i \mbox{ and } \bigwedge_{i\in I} \tau_i \leq \tau 
	\end{equation*}
	By $\IntrInter$ and $(\leq)$ rules, $\Gamma , x:\delta \der M':\tau$, and hence by Substitution lemma \ref{app:SubstLemma} $\Gamma \der M'\Subst{V}{x} : \tau$.
	\newline \newline
	Second case: $M\equiv M' \Bind \lambda x.\Unit x$ and $N=M'$ is proved in a very similar way to the first case, by applying twice Gen. Lemma and appealing $(\leq)$ rule.
	\newline \newline
	Third case: 
	$\bar{M}\equiv (L\Bind \lambda x.M)\Bind \lambda y.N \Red L \Bind \lambda x.(M\Bind \lambda y.N)\equiv \bar{N}$
	by applying Gen. Lemma \ref{app:genLemma} on $\bar{M}$:
	\[\begin{array}{lll}
	\exists I, \delta_i , \tau_i . \forall i\in I \, \Gamma \der L\Bind \lambda x. M: T\delta_i  \mbox{ and } \Gamma \der \lambda y.N: \delta_i \rightarrow \tau_i  \mbox{ and } \bigwedge_{i\in I} \tau_i \leq \tau
	\end{array}\] 
	by applying Gen. Lemma \ref{app:genLemma} on $L\Bind \lambda x. M$:
	\[\begin{array}{lll}
	\exists I_i, \delta_{ij} , \tau_{ij} . \forall i,j \, \Gamma \der L:T\delta_{ij}  \mbox{ and } \Gamma \der \lambda x.M: \delta_{ij} \rightarrow T\delta_i  \mbox{ and } \bigwedge_{ij} T\delta_{ij} \leq T\delta_i
	\end{array}\] 
	by applying Gen. Lemma \ref{app:genLemma} on $\lambda y. N$:
	\[\begin{array}{lll}
	\exists I_k, \delta_{ik} , \tau_{ik} . \forall i,k \, \Gamma,\ y:\delta_{ik} \der N:\tau_{ik}  \mbox{ and }  \bigwedge_{ik} (\delta_{ik}\rightarrow \tau_{ik}) \leq \delta_i\rightarrow \tau_i
	\end{array}\] 
	by applying Gen. Lemma \ref{app:genLemma} on $\lambda x.M$:
	\[\begin{array}{lll}
	\exists I_h, \delta_{ijh} , \tau_{ijh} . \forall i,j,h \, \Gamma,\ x:\delta_{ijh} \der M:\tau_{ijh}  \mbox{ and }  \bigwedge_{ijh} (\delta_{ijh}\rightarrow \tau_{ijh}) \leq \delta_{ij}\rightarrow T\delta_i
	\end{array}\] 
	Before moving on with the rest of the proof, we properly restate a lemma proved in \cite{BCD'83} (lemma 2.4 (ii)):
	\begin{lemma}\label{app:BCD83}
		if $\bigwedge_{i=1}^n\delta_i\rightarrow \tau_i\leq \bar{\delta}\rightarrow \bar{\tau}$ and $\bar{\tau}\neq \omega_\Comp$, then there exist $i_1,\ldots ,i_k\in\{1,\ldots ,n\}$ such that $\bigwedge_{j=1}^k\delta_{i_j}\geq \bar{\delta}$ and $\bigwedge_{j=1}^k\tau_{i_j}\leq \tau$.
	\end{lemma}
	By this lemma we extract useful inequalities:
	$\forall i\in I\, \exists \bar{I_k}\subset I_k$ and $\bar{I_h}\subset I_h$ such that:
	\[\begin{array}{lll}
	\delta_i\leq\delta_{ik}\leq\delta_{ijh} & \tau_i\geq\tau_{ik}
	\end{array}\]
	By these inequalities one can build the following derivation tree that prove this third part of the statement, the use of enforcement of basis and $(\leq)$-rule are hidden for sake of readability:
	\[\begin{array}{lll}
	
	\prooftree
	\Gamma \der L:T\delta_{ij}\leq T\delta_i \qquad 
	
	\prooftree
	
	\prooftree
	\Gamma , ,x:\delta_i\der M:T\delta_i \qquad 
	
	\prooftree
	\Gamma , ,x:\delta_i, y:\delta_i\der N:\tau_{ik}\leq \tau_i
	\justifies
	\Gamma , ,x:\delta_i \der \lambda y.N:\delta_i\rightarrow \tau_i
	\endprooftree
	\justifies
	\Gamma ,x:\delta_i\der M\Bind \lambda y.N:\tau_i
	\endprooftree
	\justifies
	\Gamma \der \lambda x.(M\Bind \lambda y.N): \delta_i\rightarrow \tau_i
	\endprooftree
	\justifies
	\Gamma \der L\Bind \lambda x.(M\Bind \lambda y.N):\tau_i\leq \tau 
	\endprooftree
	
	\end{array}\]

\QED

\subsection{The filter model construction}
\begin{lemma}\label{app:filtMonad}
	Let $\TT: |\Dcat| \to |\Dcat|$ be as above. Define $\UnitSub{D}^\Filt: \Filt_D \to \Filt_{TD}$ and 
	$\Bind_{D,E}^\Filt: \Filt_{TD} \times \Filt_{D \to \TT E} \, \to \Filt_{\TT E}$ such that:
	\[\UnitSub{D}^\Filt \ d = \Up \Set{T\delta \in \Types_{\TT D} \mid \delta \in d} \qquad
	t \Bind_{D,E}^\Filt e = \Up\Set{\tau \in \Types_{\TT E} \mid \exists\ \delta\to\tau \in e. \; T\delta \in t}\]
	Then $(\TT, \Unit^\Filt, \Bind^\Filt)$ is a monad over $\Dcat$.
\end{lemma}

\Proof
	In order to prove that  $(\TT, \Unit^\Filt, \Bind^\Filt)$ is a monad over $\Dcat$, it suffices to show that the triple respects axioms in definition \ref{def:monad}, namely:
	\begin{enumerate}
		\item $(\UnitSub{D}^\Filt \ d) \Bind_{D,E}^\Filt f = f\,d$
		\item $a \Bind_{D,D}^\Filt \UnitSub{D}^\Filt = a$
		\item $(a \Bind_{D,E}^\Filt f) \Bind_{D,E,F}^\Filt g = a \Bind_{D,E,F}^\Filt \metalambda d. (f\,d \Bind_{E,F}^\Filt g)$.
	\end{enumerate}
	\textbf{Case 1:} Let $d\in \Filt_D$ and $f\in \Filt_{D \to \TT E}$, with $D, E\in |\Dcat|$.\\
	\[
	\begin{array}{rcll}
	(\UnitSub{D}^\Filt \ d) \Bind_{D,E}^\Filt f & = & \Up \Set{\tau \in \Types_{\TT E} \mid \exists\ \delta\to\tau \in f. \; T\delta \in \UnitSub{D}^\Filt \ d}  & \mbox{by definition}\\ [1mm]
	& = & \Up \Set{\tau \in \Types_{\TT E} \mid \exists\ \delta\to\tau \in f. \; T\delta \in \Up \Set{T\delta' \in \Types_{\TT D} \mid \delta' \in d}} & \mbox{by definition}\\[1mm]
	& = & \Up \Set{\tau \in \Types_{\TT E} \mid \exists\ \delta\to\tau \in f. \; \delta \in d}\\[1mm]
	\end{array}
	\] 
	Set $\Set{\tau \in \Types_{\TT E} \mid \exists\ \delta\to\tau \in f. \; \delta \in d}=:A$.\\
	Since \[f \cdot d = \Set{\tau \in \Types_E \mid \exists\, \delta \to \tau \in f. \;\delta \in d} \in \Filt_E\], one has to prove that $\Up A =A$, and, in particular, to show that  if $\tau \in A$ then $\tau' \in A \forall \tau'\geq \tau$.\\
	In fact, if $\tau'\geq_{\TT E} \tau$ then $\delta\to\tau\leq_{D\to \TT E}\delta \to \tau'$, and if $\delta \to \tau\ in f$ then also $\delta\to\tau'\in f$, because $f$ is a filter.\\
	\newline
	\textbf{Case 2:} Let $a\in \Filt_{\TT D}$ and $\UnitSub{D}^\Filt\in \Filt_{D \to \TT D}$. Actually, $\UnitSub{D}^\Filt \in \Filt_D\to\Filt_{\TT D}$ that is isomorphic to $\Filt_{D \to \TT D}$ by Proposition \ref{prop:eats}.
	\[
	\begin{array}{rcll}
	a \Bind_{D,D}^\Filt \UnitSub{D}^\Filt & = & \Up \Set{\tau \in \Types_{\TT D} \mid \exists\ \delta\to\tau \in \UnitSub{D}^\Filt. \; T\delta \in a}  & \mbox{by definition}\\ [1mm]
	\end{array}
	\]
	But if $\delta \to\tau\in \UnitSub{D}^\Filt$ then $\tau=_{\Comp}T\delta$
	\[
	\begin{array}{rcll}
	a \Bind_{D,D}^\Filt \UnitSub{D}^\Filt & = & \Up \Set{\tau \in \Types_{\TT D} \mid T\delta \in a} & = a\\[1mm]
	\end{array}
	\]
	\textbf{Case 3:} By direct, but tedious, calculations.
\QED

\subsection{Soundness and completeness of the type system}

\begin{lemma}\label{app:inclusion}
	Let $D$ be a $T$-model,
	$\Sem{\cdot}^{TD}$ be a monadic type interpretation, and assume $\xi \in \TypeEnv_D$.\\
	The couple $(D,\xi)$ \emph{preserves} $\leq_\Val$ and $\leq_\Comp$, that is: for all $\delta, \delta' \in \ValType$ and for all $\tau, \tau'\in \ComType$, one has:
	\[ \delta \leq_\Val \delta' ~ \Then ~ \Sem{\delta}^D_\xi \subseteq \Sem{\delta'}^D_\xi 
	\quad \mbox{and} \quad
	\tau \leq_\Comp \tau' ~ \Then ~ \Sem{\tau}^{TD}_\xi \subseteq \Sem{\tau'}^{TD}_\xi 
	\]
\end{lemma}
\Proof
	The statement is proved once every axiom of (pre-)order defined in \ref{def:type-theories-Th_V-Th_C} on type theories $\Th_\Val$ and $\Th_\Comp$ is preserved by the interpretation assignment.
	\[ 
	\delta \leq_\Val \omega_\Val \qquad \tau \leq_\Comp \omega_\Comp
	\]
	In these cases, the statement is easily proved, as $\Sem{\omega_\Val}^D_\xi=D$ and $\Sem{\omega_\Comp}^{TD}_\xi=TD$.
	\[\omega_\Val \leq_\Val \omega_\Val \to \omega_\Comp
	\]
	The interpretations are, respectively, $\Sem{\omega_\Val}^D_\xi=D$ and $\Sem{\omega_\Val \to \omega_\Comp}^D_\xi =\Set{d\in D\mid \forall d' \in \Sem{\omega_\Val}^D_\xi d\cdot d'\in \Sem{\omega_\Comp}^{TD}_\xi}=\Set{d\in D\mid \forall d'\in D\; d\cdot d'\in TD}$. By these analytic descriptions one deduces that $\Sem{\omega_\Val}^D_\xi = \Sem{\omega_\Val \to \omega_\Comp}^D_\xi$ and the statement is \textit{a fortiori} true.
	\[
	(\delta \to \tau) \Inter (\delta \to \tau') \leq_\Val \delta \to (\tau \Inter \tau')
	\]
	By the definition of type interpretation, one has $\Sem{(\delta \to \tau) \Inter (\delta \to \tau')}^D_\xi=\Sem{\delta \to \tau}^D_\xi \cap \Sem{\delta \to \tau'}^D_\xi = \Set{d\in D\mid \forall d'\in \Sem{\delta}^D_\xi \; d\cdot d'\in \Sem{\tau}^{TD}_\xi \mbox{ and } d\cdot d'\in \Sem{\tau'}^{TD}_\xi}$.
	So, also in this case, the two interpretations are equal.\\
	\[
	\prooftree
	\delta' \leq_\Val \delta \quad \tau \leq_\Comp \tau'
	\justifies
	\delta \to \tau \leq_\Val \delta' \to \tau'
	\endprooftree 
	\]
	By induction hypothesis $\Sem{\delta'}^D_\xi \subseteq \Sem{\delta}^D_\xi$ and $\Sem{\tau}^{TD}_\xi \subseteq \Sem{\tau'}^{TD}_\xi$. By definition, $\Sem{\delta \to \tau}^D_\xi = \Set{d\in D\mid \forall d'\in \Sem{\delta}^D_\xi\; d\cdot d'\in \Sem{\tau}^{TD}_\xi}$. Thus, if $\bar{d} \in \Sem{\delta \to \tau}^D_\xi$ this means that $\forall d'\in \Sem{\delta}^D_\xi\; \bar{d}\cdot d'\in \Sem{\tau}^{TD}_\xi$ and, since $\Sem{\delta'}^D_\xi \subseteq \Sem{\delta}^D_\xi$,  we have $\forall d'\in \Sem{\delta'}^D_\xi\; \bar{d}\cdot d'\in \Sem{\tau}^{TD}_\xi$. Finally, since $\Sem{\tau}^{TD}_\xi \subseteq \Sem{\tau'}^{TD}_\xi$, $d\in \Sem{\delta' \to \tau'}^D_\xi$.
	\[ 
	\prooftree
	\delta \leq_\Val \delta'
	\justifies
	T\delta \leq_\Comp T\delta'
	\endprooftree\\ 
	\]
	We know by inductive hypothesis that $	\delta \leq_\Val \delta'$ implies that $\Sem{\delta}^D_\xi \subseteq \Sem{\delta'}^D_\xi$
	and hence the thesis follows by monadicity.
	\\ For what concerns the remaining case $	T\delta \Inter T\delta' \leq_\Comp T(\delta \Inter \delta') $, by explicitly writing the analytic description of each interpretation, one finds out that the two sets are equal. 
\QED

\end{document}